%% file: stochsobol.tex
\documentclass[final]{siamltex1213hacked}
\usepackage{amsmath,amssymb, amsfonts}
\usepackage{algorithm}
\usepackage{float}
\usepackage{mathtools}
\usepackage{subcaption}
 \usepackage{relsize}
\usepackage{algpseudocode}
\usepackage{graphicx}
\usepackage{booktabs} 
\usepackage[table]{xcolor}
\usepackage{bm}

\title{Efficient computation of Sobol' indices for stochastic models\thanks{This work was
supported by the National Science Foundation under grant DMS-1522765.}}

\author{J.~L. Hart\thanks{Department of Mathematics, North Carolina State
University, Raleigh, NC 27695-8205 \email{jlhart3@ncsu.edu}}  \and A.
Alexanderian\thanks{Department of Mathematics, North Carolina State University,
Raleigh, NC 27695-8205 \email{alexanderian@ncsu.edu}} \and P.~A.
Gremaud\thanks{Department of Mathematics, North Carolina State University,
Raleigh, NC 27695-8205 \email{gremaud@ncsu.edu}} }

\input{def}

\begin{document}
\maketitle
\newcommand{\slugmaster}{%
\slugger{juq}{xxxx}{xx}{x}{x--x}}%slugger should be set to juq, siads, sifin, or siims

\renewcommand{\thefootnote}{\fnsymbol{footnote}}

\begin{abstract}
 Stochastic models are  necessary for the realistic description of an increasing number of applications. 
The ability to identify influential parameters and variables is critical to 
 a thorough analysis and understanding of the underlying phenomena.   
 We present a new global sensitivity analysis approach for stochastic models, i.e., models with both uncertain parameters and intrinsic 
 stochasticity. Our method relies on an analysis of variance 
 through a generalization  of Sobol' indices and on the use of surrogate models. We show how to efficiently compute  the statistical properties of the resulting indices and illustrate the effectiveness of our approach by computing first 
order Sobol' indices for two stochastic models.
\end{abstract}



\pagestyle{myheadings}
\thispagestyle{plain}
\markboth{J.L. HART, A. ALEXANDERIAN,  AND P.A.~GREMAUD}{Sobol' indices for stochastic models} 

\section{Introduction}
\label{sec:intro}

Stochastic computer models are non-deterministic simulators: repeated evaluations with the same inputs yield different outputs.  Examples
include agent-based models, queuing models, Monte-Carlo based numerical models
and models of intrinsically stochastic phenomena such as those found in
biological systems \cite{darren} or chemical reaction networks
\cite{ssa1}. We consider stochastic computer models of the form
\begin{equation}
   Y = f(\mathbf X(\vartheta), \omega),   \label{stochmod}
\end{equation}
where  $\mathbf X= (X_1, \dots, X_p)$ is a random vector whose entries
are uncertain model parameters;
the variables $\vartheta$ and $\omega$ correspond to  two distinct
sources of randomness, namely the uncertain parameters in the model and  the
stochasticity of the problem, respectively. The precise mathematical formulation of (\ref{stochmod}) 
is given in  Section~3. 

The development and predictive capabilities of such computer models
depend on the ability to apportion uncertainty in the model output to different sources of uncertainty in the model 
input parameters and intrinsic stochasticity, i.e., on global sensitivity analysis \cite{saltellibook}.  To that end, efficient 
methods \cite{iooss,sobolowen,saltellibook,sobol93,sobol} have been developed for the
 simpler model 
\begin{equation}
Y = f(\mathbf X(\vartheta)), \label{detmod}
\end{equation}
which only incorporates parametric uncertainty. 
For stochastic models such as~\eqref{stochmod},  however,  even the 
concept of sensitivity is delicate.  Additionally,  stochastic models are both
computationally more demanding and substantially harder to fit to data than
their deterministic counterparts; the need for efficient and reliable
sensitivity analysis in the context of (\ref{stochmod}) is thus clear.

In this article, we propose a new notion of global sensitivity for stochastic models
based on
\begin{itemize}
\item  a generalization of Sobol' indices  \cite{sobol93}  to the case of stochastic models,
\item the use of  surrogate models.
\end{itemize}
We briefly introduce both concepts. 

For (\ref{detmod}), the Sobol' indices are defined as
\begin{equation}
S_u = \frac{\operatorname{Var}\{ \mathbb E\{f(\mathbf X)| \mathbf X_u\}\} }{\operatorname{Var}\{f(\mathbf X)\}} , 
\qquad u \subset \{1,2,\dots,p\}, \label{Sobolindex}
\end{equation}
where $\mathbf X_u$ denotes the subset of entries in $\mathbf X$ corresponding to $u$; for instance $\mathbf X_{\{2,5\}}=(X_2,X_5)$. The indices apportion relative contributions to the variance of
the output among the inputs; variables contributing more (larger
$S_u$'s) are deemed more important. When $u=\{k\}$, $S_k$ is called the first order Sobol' index; when $u=\{k\}^c$, $T_k=1-S_u$ is called the total Sobol' index. A direct application of this concept 
to~\eqref{stochmod} instead of~\eqref{detmod}  yields Sobol' indices
$S_u$, $u \subset \{1,2,\dots,p\}$, which are themselves random variables. Example~\ref{ex:synthetic} 
below illustrates this point; a full justification is given in Section~3.

Traditional methods to evaluate the Sobol' indices  (\ref{Sobolindex}) involve Monte Carlo integration \cite{saltelli} and are  infeasible for problems where $f$ is expensive to evaluate.
To overcome this obstacle, a surrogate model $\hat f$ can be constructed whereby
\begin{itemize}
\item  $\hat f$ is representative of $f$, i.e., $\hat f \approx f$ in some sense, 
\item  $\hat f$ can be evaluated cheaply.
\end{itemize}
Several families of surrogate models  (or metamodels) have been proposed
including polynomial chaos expansions~\cite{Sudret:2008,Crestaux:2009,kdn,LeMaitreKnio10,BlatmanSudret10,GratietMarelliSudret15},
Kriging models and Gaussian processes \cite{kleijnen,legratiet,GratietMarelliSudret15}, and
non-parametric statistical models \cite{hastie2009,sacks}.

Our proposed approach is as follows: 

\begin{enumerate}
\item construct a surrogate model $\hat{f}$ of (\ref{stochmod}),
\item compute the Sobol' indices of $\hat{f}$ (which are here random variables themselves),
\item compute the statistical properties of the Sobol'  indices.
\end{enumerate}

While replacing $f$ by $\hat f$ greatly facilitates computational analysis, it also 
creates a fundamental difficulty: to what extend is the global sensitivity analysis of $\hat f$ reflective of the properties 
of $f$? This difficult and general question is largely open in the context of global sensitivity analysis (see, however, our discussion in Section~2.1 below). 

For stochastic models such as~\eqref{stochmod}, an alternative  approach is to
first marginalize over $\omega$, and then evaluate the Sobol' indices. For instance, 
we can construct a surrogate model $\hat g$ for the expected value in $\omega$
\[
\hat{g}(\vec{X}(\vartheta)) \approx g(\vec{X}(\vartheta)) \defeq \mathbb E_\omega \{f(\vec{X}(\theta),\omega ) \},
\]  
and then evaluate the Sobol' indices of $\hat g$.  This method is  used for
example in  \cite{IoossRibatet09,marrel,marrel3}.  In other words, the same
three steps noted above are used, but in a different order: 3, 1, 2. At the heart  of our
approach is the fact that, for appropriate surrogates, it is possible to
efficiently and directly compute  sensitivity information for the stochastic
model (\ref{stochmod}) without such a priori marginalizations.  Moreover,  {\em
computing $\omega$-moments} and {\em evaluating Sobol' indices}
for~\eqref{stochmod} are two operations that do not commute.  This simple
observation has significant consequences as averaging over $\omega$ before
computing sensitivity indices significantly reduces the amount of information
available for  analysis. This point is illustrated by the following example.

\begin{example}
\label{ex:synthetic}
Let  $(\Theta,\mathcal{E}, \lambda)$, $(\Omega, \mathcal{F}, \nu)$, be probability spaces, 
and $\mathbf X:\Theta \rightarrow \R^2$, 
$W:\Omega \rightarrow \mathbb R$ be random variables defined as follows.
We let $\mathbf X(\vartheta) = (\mu(\vartheta), \sigma(\vartheta))$ with $\mu \sim \mathcal U(0,1)$,
 $\sigma \sim \mathcal U(1,L+1)$ for some positive $L$ and
$W \sim \mathcal N(0,1)$, where $\mathcal U$ and $\mathcal N$ denote uniform and normal distributions, respectively.  
We consider an example of a stochastic model of the form~\eqref{stochmod} as follows:
\begin{eqnarray}
Y = f(\mathbf X(\vartheta),\omega) = \mu(\vartheta) + \sigma(\vartheta) W(\omega),  \label{toyprob}
\end{eqnarray}
For $L=0$, $\sigma$ is deterministic;
as the value of $L$ increases, so does the uncertainty on $\sigma$. We
therefore expect the importance of $\sigma$ to increase with $L$. This is
confirmed by direct calculations. The first order Sobol' indices of $Y$ with
respect to both $\mu$ and $\sigma$ can be found analytically 
\[
S_\mu(Y)(\omega) = \frac 1{1+L^2W(\omega)^2} \quad \mbox{ and }\quad  S_{\sigma}(Y)(\omega) =  \frac {L^2W(\omega)^2}{1+L^2W(\omega)^2},
\]
and the corresponding expected values are given by 
\[
\begin{aligned}
\mathbb E_{\omega}\{S_\mu(Y)\} &= 
\frac 1L \sqrt{\frac {\pi}2} \exp\left( \frac 1{2L^2}\right) \operatorname{erfc}\left(\frac 1{\sqrt{2}L}\right),   \\ 
\mathbb E_{\omega}\{S_\sigma(Y)\}  &= 1 - \mathbb E_{\omega}\{S_\mu(Y)\},
\end{aligned}
\]
where $ \operatorname{erfc}$ is the complementary error function. 
Figure~\ref{fig:toyprob} illustrates the behavior of $\mathbb E_{\omega}\{S_\sigma(Y)\} $  as a function of $L$ confirming the increasing importance of $\sigma$ with $L$. 
\begin{figure}[ht]
\centering
\includegraphics[width=.65\textwidth]{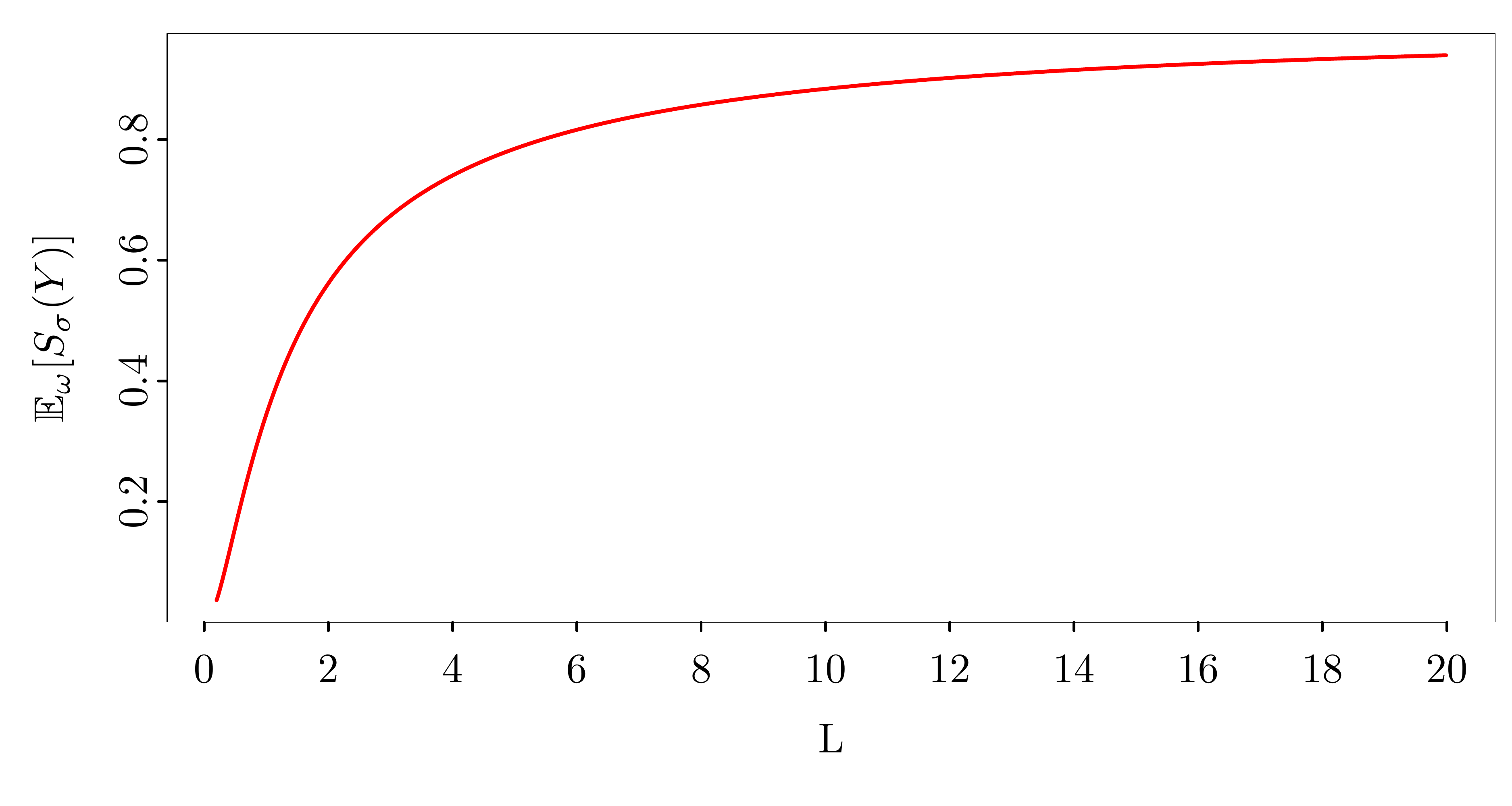}
\caption{Expected first order Sobol' index of $Y$ from (\ref{toyprob}) with respect to the uncertain parameter $\sigma$ as a function  of $L$.}
\label{fig:toyprob}
\end{figure}

Reversing the order of operations between averaging and computing the Sobol' indices leads
to an entirely different picture which is at odds with the very nature of (\ref{toyprob}). Indeed, the expected value of $Y$ with
respect to $\omega$ is simply $\mathbb E_{\omega}\{Y\} = \mu$ and therefore the
first order Sobol' indices are given by 
\[
S_{\mu}(\mathbb E_{\omega}\{Y\}) = 1 \quad \mbox{ and } \quad S_{\sigma}(\mathbb E_{\omega}\{Y\}) = 0.
\]
In other words, $\sigma$ appears insignificant regardless of $L$. 
\end{example}

A fair amount of recent work on  global sensitivity analysis for stochastic models has been directed toward 
the analysis of stochastic chemical systems \cite{DegasperiGilmore08,LeMaitreKnio15,lemaitrekm}. 
 In~\cite{LeMaitreKnio15}, for instance,  the authors develop a method, based on
polynomial chaos expansion and a stochastic Galerkin formalism, for
the analysis of variance of stochastic differential equations driven by additive
or multiplicative Wiener noise. In \cite{lemaitrekm},  a method is proposed for  the sensitivity analysis of
stochastic chemical systems to the different reaction channels and channel
interactions.   There has also been significant progress in \emph{local} (derivative-based)
sensitivity analysis for stochastic systems.  References~\cite{gunawan,kdn,nakayama,plyasunov,rathinam} provide a 
sample of such efforts.  Finally, other approaches for sensitivity analysis of
stochastic systems rely on  information
theory~\cite{aram,komo,majda1,majda2}. 

Surrogate models are an important component of our approach.  While our framework is agnostic to the choice of surrogates,
practical considerations such as ease of calculation and efficiency in high dimensions have to be taken 
into account. For many surrogates, the Sobol' indices can be evaluated at negligible cost or even
analytically. Further, since by construction  
\[
0\le S_u \le 1, \mbox{ for } u \subset \{1,2,\dots, p\}, 
\]
the moments of the Sobol' indices can be computed efficiently, as shown in 
Section~\ref{sec:formulation}. Therefore, the most expensive of the steps
1--3 mentioned above
is the construction of  the surrogate model itself. 

We rely on Multivariate Adaptive Regression Splines 
(MARS)~\cite{friedman91,friedman93,hastie2009} for surrogate construction. 
MARS is a nonparametric model which adaptively allocates basis functions. This
results in the surrogate itself
screening variables prior to the application of  a more sophisticated tool from sensitivity
analysis (such as Sobol' indices). MARS approximations tend to omit the less important variables. 
Consequently,  the Sobol' indices of the influential variables are ``biased high" while the indices of the 
less influential variables are ``biased low". This apparent flaw is in fact a benefit for global sensitivity analysis as
the accurate identification of influential vs non-influential variables is  the goal; the Sobol' indices 
are but a tool to obtain that information. We demonstrate the efficiency of MARS in the context of global sensitivity analysis 
\cite{storlie} by comparing it to polynomial chaos   \cite{Ghanem:1991a,LeMaitreKnio10}, see Section 2.

The proposed method is tested on two numerical examples  in Section~\ref{sec:numerical}. The first example is a synthetic problem based
on a stochastic version of the well-known g-function~\cite{sobol}. Inexpensive function evaluations 
and analytic expressions for the Sobol' indices facilitate the  systematic
assessment of the method. The second example
involves a stochastic biochemical reaction network exhibiting fast timescales
and an oscillatory behavior.    We illustrate the performance of our method by
estimating the oscillatory time dependent behavior of the Sobol' indices.

\section{Surrogate models for sensitivity analysis}
\label{sec:surrogate}
In this section, we focus on first order Sobol' indices and consider their
computation using surrogate models; similar analysis may be done for higher
order Sobol' indices. 

\subsection{Accuracy}

Let $f$ be as in~\eqref{detmod} and $\mathbf{S}=(S_1,
S_2, \dots, S_p)$ be its first order Sobol' indices.  Further, let
$\hat{\mathbf{S}}=(\hat{S}_1,\hat{S}_2,\dots, \hat{S}_p)$ be the corresponding
indices for a surrogate $\hat{f}$. Ideally, $\mathbf{S}$ and $\hat{\mathbf S}$
would lead to the identification of the same set of influential  variables;
various metrics can be considered for computing the discrepancy between these
two vectors of indices. It is often observed in practice that $\sum_{k=1}^p
\hat{S}_k>\sum_{k=1}^p S_k$. As the Sobol' indices measure the {\em relative}
importance of the variables in any given problem, we define  a corresponding
error $E$ through the following normalization

\begin{equation}\label{errordef}
E=\left\| \frac{\mathbf S}{\sum\limits_{k=1}^p S_k}-\frac{\hat{\mathbf S}}{\sum\limits_{k=1}^p \hat{S}_k}\right\|_{\infty}, 
\end{equation}
and  refer to ${\mathbf S}/(\sum_{k=1}^p S_k)$ as the normalized indices.  The
choice of the norm, here $\ell^{\infty}$, has little effects on the results
presented in this paper;  using the $\ell^1$ and $\ell^2$ norms lead to similar
conclusions. We are not aware of theoretical results regarding either $E$ from
(\ref{errordef}) or other similar error measures;  error assessment is however
discussed in \cite{janon14} for surrogates admitting local error bounds (which
is not the case of most methods from non-parametric statistics including MARS).

\subsection{MARS}

Let $f$ and $\vec{X}$ be as in~\eqref{detmod}. We assume $f \in L^2(\mathcal{X}, \mathcal{B}(\mathcal{X}), F_\vec{X})$,
where $F_\vec{X}$ denotes the distribution function of $\vec{X}$,
$\mathcal{X} \subseteq \R^p$ is the support of the distribution law of $\vec{X}$, and $\mathcal{B}(\mathcal{X})$ is the Borel sigma-algebra on $\mathcal{X}$. MARS \cite{friedman91,friedman93} approximations to $f$ are constructed through an
adaptive regression procedure involving truncated one-sided linear splines and
products thereof. More precisely, let 
\[
\mathcal C = \{ (x_j - t)_+, (t - x_j)_+ : t \in \{x_{i,j}\}, i = 1, \dots , n, j =1, \dots, p\}
\]
be a set of $2np$ elementary functions (assuming distinct input values), where
$\{x_{i,j}\}$ is the set of available data and, for $\xi \in \R$, $\xi _+ = \max\{0,\xi\}$.
The model is of the form 
\begin{eqnarray}
\hat f(\mathbf x ) = \beta _0 + \sum_{m=1}^M \beta _m \, \phi_m(\mathbf x), \label{mars}
\end{eqnarray}
where  $\Phi = \{ \phi_1,\phi_2,...,\phi_M \}$ is the basis (constructed in
algorithm \ref{MARSbasis}) and the $\beta_m$'s are obtained through standard
linear regression.

\begin{algorithm}
\caption{MARS basis}\label{MARSbasis}
\begin{algorithmic}
   \State $\Phi = \{1\}$
   \While{$|\Phi| \le \mbox{max size}$ (\mbox{max size} $> M$)}
      \State{find $(\ell ^\star, j^\star, i^\star)$ corresponding to the best approximation of the form}
      \State{$\hat{f}(\mathbf x) \in \mbox{span}\{\Phi,\phi_\ell(\mathbf x)(x_j-t)_+, \phi_\ell(\mathbf x)(t-x_j)_+$ : $\phi_\ell \in \Phi$, $t\in\{x_{i,j}\} \}$}
      \State{$\Phi = \Phi \cup \{ \phi_{\ell^\star}(\mathbf x)(x_{j^\star}-x_{i^\star,j^\star})_+,\phi_{\ell^\star}(\mathbf x)(x_{i^\star,j^\star}-x_{j^\star})_+ \}$}
   \EndWhile
   \State{$\Phi = \{\phi_1,\phi_2,...,\phi_{|\Phi|} \}$}
   \While{$|\Phi| > M$}
      \State{find j corresponding to the best approximation of the form}
      \State{$\hat{f}(\mathbf x) \in \mbox{span}\{ \Phi \setminus \{ \phi_j \} \} $}
      \State{$\Phi=\Phi \setminus \{ \phi_j \} $}
   \EndWhile
   \vspace{3 mm}
   \State Note: $M$ is chosen by the algorithm, not  by the user. Statistical tests are used to determine when to end the while loops \cite{friedman91}. We use $M$  to simplify the presentation of the algorithm which is more complex in its actual implementation.
   \end{algorithmic}
\end{algorithm}

MARS is often used with the lowest degree of interaction  \cite{earth}, namely one,
in which case it  corresponds to an additive model. This is the approach we
adopt below. We use the R function {\sc{earth}} \cite{earth} to build MARS
surrogates. The additive structure of a MARS surrogate
with degree of interaction one enables analytic computation of the Sobol'
indices.  As we work exclusively with additive MARS surrogates,
we  focus on first order Sobol' indices--higher order Sobol' indices would carry no additional information \cite{sobol93}.  
However,  higher order and total Sobol' indices could also be considered
in the proposed framework provided the  surrogate model
incorporates mixed terms, i.e.,  interactions between different
uncertain parameters.

The additive MARS model can be represented as
\begin{equation}\label{marsadd}
\hat f(\mathbf x ) = \beta _0 + \sum_{k=1}^p \sum_{j=1}^{M_k} \beta _{k,j} \, \phi_{k,j}(x_k), 
\end{equation}
where $M_k$
is the number of basis functions depending on $x_k$.  
Let $I_{k,j}=\int \phi_{k,j}dF_{\mathbf x}$ be the mean of $\phi_{k,j}$; we can then rewrite (\ref{marsadd})  as

\begin{eqnarray}
\hat f(\mathbf x ) = \left(\beta _0+\sum_{k=1}^p \sum_{j=1}^{M_k} \beta_{k,j}I_{k,j}\right) + \sum_{k=1}^p \sum_{j=1}^{M_k} \beta _{k,j} \, (\phi_{k,j}(x_k)-I_{k,j}), \label{marsanova}
\end{eqnarray}
which is the ANOVA decomposition of $\hat f$. The first order Sobol' indices
are obtained analytically by computing $V_k=\int\left( \sum_{j=1}^{M_k} \beta
_{k,j} \, (\phi_{k,j}(x_k)-I_{k,j}) \right)^2dF_{\mathbf x}$ for $k=1,\dots p$
and setting $S_j={V_j}/{(\sum_{k=1}^p V_k)}$.

\subsection{A numerical example}
In this subsection we demonstrate the utility of MARS to compute Sobol' indices
for problems of the form~\eqref{detmod}.  To this end, we compare results
obtained using MARS against those computed using a polynomial chaos (PC)
expansion, which is a well known tool for constructing surrogate models. 

Before presenting the numerical test, we briefly recall some basics regarding
PC expansions.  The PC expansion of $f \in L^2(\mathcal{X},
\mathcal{B}(\mathcal{X}), F_\vec{X})$ is a series expansion of the type $f =
\sum_{k = 0}^\infty c_k \Psi_k$, where $\{\Psi_k\}_0^{\infty}$ is a set of
$p$-variate polynomials forming an orthogonal basis of $L^2(\mathcal{X},
\mathcal{B}(\mathcal{X}), F_\vec{X})$. The PC basis is dictated by the
statistical distribution  of the uncertain parameters $X_1, \ldots, X_p$.  For
example, if $X_1, \ldots, X_p$ are iid uniform random variables,  the PC basis
can be taken as $p$-variate Legendre polynomials. 
Implementation is done through  truncated expansions of the form \newcommand{\Npc}{{n_\text{pc}}}
\begin{equation}\label{equ:PCE}
f \approx \sum_{k = 0}^{\Npc} c_k \Psi_k,
\end{equation}
where the number of retained basis functions $\Npc$
depends on the truncation strategy. For instance,  the case of basis functions of total order not exceeding $r$
results in 
$\Npc = (p + r)! / (p! r!)$.  

Computation of PC coefficients can be a difficult problem for computationally
extensive models~\cite{LeMaitreKnio10,Xiu10}. This has led to development
of various efficient approaches for computing PC expansions for computationally
intensive mathematical models in recent years; 
see e.g.,~\cite{BlatmanSudret11,DoostanOwhadi11,ConradMarzouk13,WinokurKimBisettiEtAl14,
JakemanEldredSargsyan15}.

For the numerical illustrations below, we compute the PC expansion through  a regression
based method that encourages sparsity  by controlling the $\ell_1$ norm of the
PC coefficient vector
\begin{equation}\label{equ:optim}
   \min_{\vec{c} \in \R^{\Npc}} \| \boldsymbol{\Lambda} \vec{c} - \vec{d} \|^2, \\ 
      \qquad \text{subject to }  \frac{1}{\Npc} \sum_{k = 0}^\Npc |c_k| \leq \tau,
 \end{equation}
 where $\boldsymbol{\Lambda} \in \R^{n\times\Npc}$, $\Lambda_{jk} = \Psi_k(\vec{X}_j)$, and $\vec{d} = \big(f(\vec{X}_1), \ldots, f(\vec{X}_n)\big)$. 
We use the solver~\cite{spgl1:2007} for the
solution of the above optimization problem, with $\tau = 0.025$, and compute a
third order PC expansion.  For the purposes of sensitivity analysis, once a PC
expansion is available, the Sobol' indices can be computed analytically
\cite{AlexanderianEtAl12,Crestaux:2009,Sudret:2008}.

For our comparison we consider the classical $g$-function initially proposed by
I. Sobol' in \cite{sobol};  this corresponds to the synthetic function
(\ref{sgfunction}) of Section~\ref{sec:numerical} with the random parameters
replaced by their expected values, i.e., 

\begin{eqnarray}
f(\vec{X}(\vartheta)) = \prod\limits_{k=1}^{15} \frac{|4X_k(\vartheta)-2|+\mathbb E_{\omega}\{a_k\}}{1+\mathbb E_{\omega}\{a_k\}},  \label{detgfun}
\end{eqnarray}
where the $a_k$'s are given in Table~\ref{tab:aj}. We construct MARS and PC surrogates for sensitivity analysis by 
sampling the model $n$ times, through a Latin Hypercube design, with 
$n=100,150,200,\dots, 950,1000$. Because of randomness in the data
sampling, the experiment is repeated 500 times for each fixed $n$ and the
errors (\ref{errordef}) are averaged. 
Figure~\ref{fig:MARSPCE} (left) displays  the resulting average error $\bar{E}$
as a function of $n$.  These errors can be interpreted by considering
Figure~\ref{fig:MARSPCE} (right) in which we show the normalized exact indices
alongside their normalized MARS and  PC approximations for a representative sample of
size  $n=600$. 
\begin{figure}[ht]\centering
\includegraphics[width=.46\textwidth]{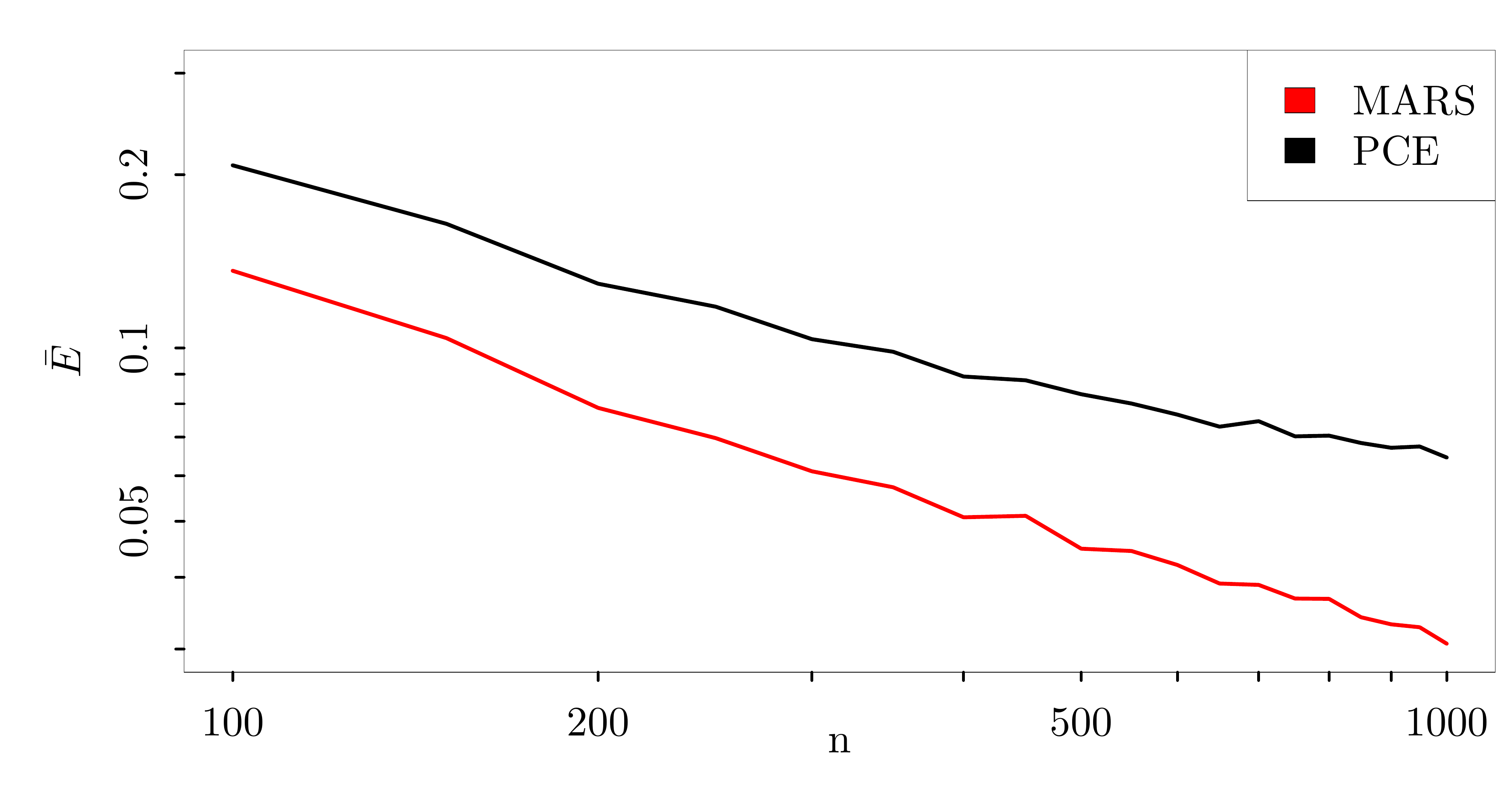}
\includegraphics[width=0.46\textwidth]{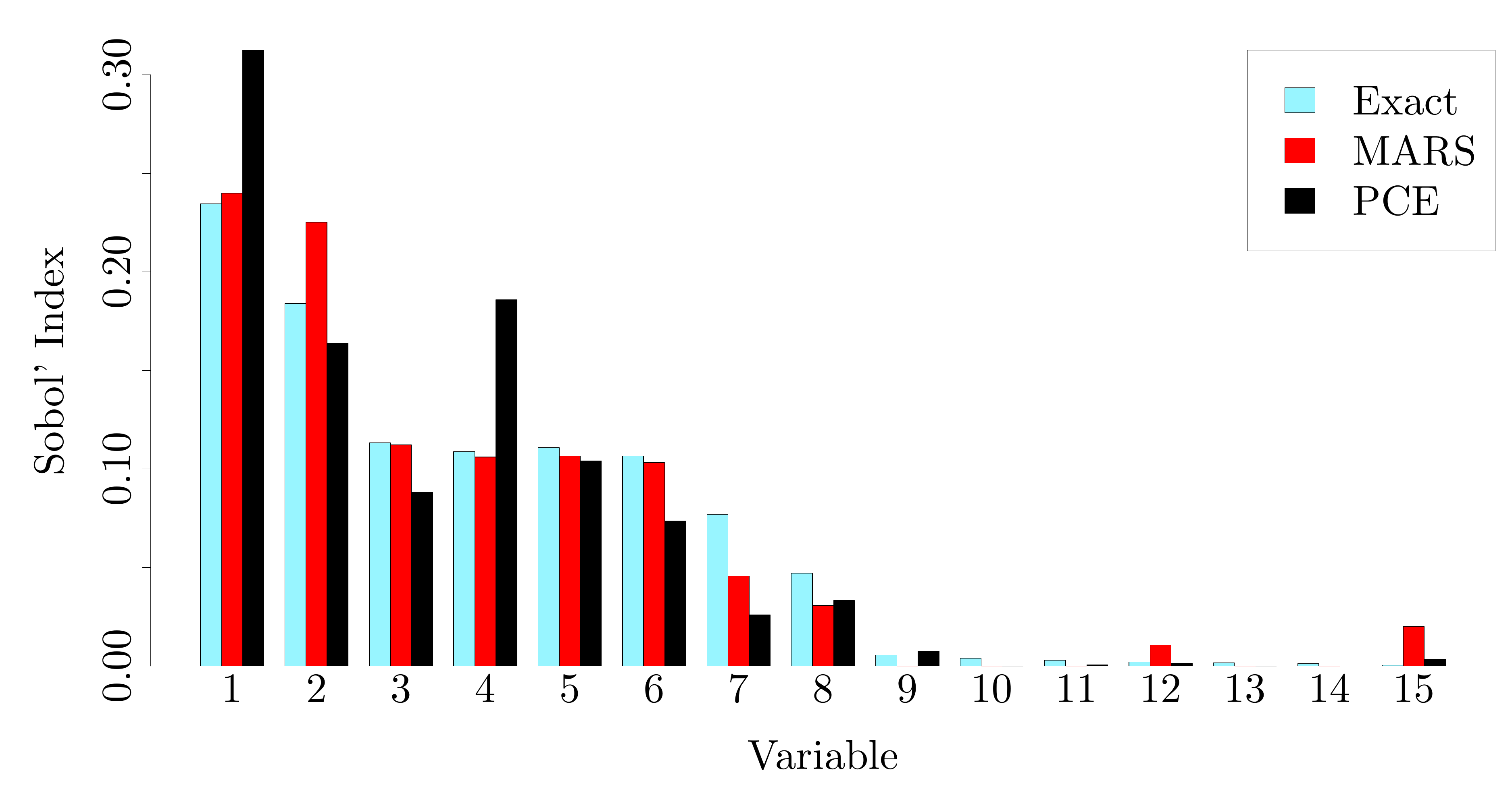}
\caption{ 
Left: average errors from   \eqref{errordef} for the Sobol indices.
The empirical convergence rates 
are 0.65 for MARS and 0.51  for PC.
Right: bar plots comparing the normalized exact indices with their normalized MARS and  PC approximations for a representative sample of size $n=600$.
}
\label{fig:MARSPCE}
\end{figure}

We observe that in the case of (\ref{detgfun}) and  for the implementations
described above, MARS and PC perform comparably for the purpose of global
sensitivity analysis when provided an equal number of function evaluations.  In
particular,  the average error improves for both methods as the sample size is
increased. We note that, in the present example, MARS is found
to have a slight edge in terms of flexibility and accuracy, when compared to
PC-based computation of Sobol indices using the above implementations.

%MARS is found to have a slight edge in terms
%of flexibility and accuracy.

\section{Formulation and Method}
\label{sec:formulation}

We start by providing a precise mathematical definition of
the function $f$ in~\eqref{stochmod}.
Let $(\Theta, \mathcal{E}, \lambda)$
be the probability space  associated with the uncertain parameters in (\ref{stochmod}) and $\vec{X}:\Theta \to
\R^p$ be the corresponding random vector. In addition, we consider another
probability space, $(\Omega, \mathcal{F}, \nu)$, that carries the stochasticity
of a model. 
The corresponding product space 
\[
(\Theta, \mathcal{E}, \lambda) \otimes (\Omega,
\mathcal{F}, \nu) = (\Theta \times \Omega, \mathcal{E} \otimes \mathcal{F},
\lambda \otimes \nu),
\]
 can be constructed in a standard fashion: $\mathcal{E} \otimes \mathcal{F}$ is the product
$\sigma$-algebra and $\lambda \otimes \nu$ is the product measure. 
We let $f$ be a function defined on $\mathcal{X} \times \Omega$, where 
$\mathcal{X} \subseteq \R^p$ is the support of the distribution law of
$\vec{X}$. For the type of stochastic problems we consider here, the response functions  are of the form
$Y(\vartheta, \omega) = f(\vec{X}(\vartheta), \omega)$.  We assume $Y:\Theta
\times \Omega \to \R$ belongs to $L^2(\Theta \times \Omega, \mathcal{E} \otimes
\mathcal{F}, \lambda \otimes \nu)$.
For a fixed $\omega \in \Omega$, we may consider $f(\vec{X}(\cdot),\omega) : \Theta
\to \mathbb R$ as a deterministic function of uncertain parameters and compute the Sobol' indices \eqref{Sobolindex} for each $\omega \in \Omega$. This defines the functions
\begin{eqnarray*}
S_u: \Omega &\to& [0,1], \\
\omega &\mapsto & S_u(f(\vec{X}(\cdot),\omega)), \qquad u \subset \{1,2,\dots,p\}.
\end{eqnarray*}
Invoking the remarks in Theorem 1.7.2 of \cite{durrett} and elementary properties of 
measurable functions, it can be  seen that the $S_u$'s   are
$\mathcal{F}$-measurable functions, i.e., they are random variables. 

As illustrated by Example~\ref{ex:synthetic}, the distribution of the indices
$\{S_u(\omega)\}_{u \subset \{1,2,\dots,p\} }$ may contain significant information needed for sensitivity
analysis. Computing these sensitivity indices is, in general,  costly. For
instance, to compute all $p$ first order indices $\{S_k(\omega)\}_{k=1}^p$ through the sampling based
method from \cite{saltelli} with $N$ Monte Carlo samples  requires $(p+1)N$
evaluations of $f(\vec{X}(\cdot),\omega)$ for each fixed $\omega \in \Omega$.
To characterize the statistical properties of the  indices, an additional Monte
Carlo sampling over $\Omega$ has to be performed.  Assuming a sample size of
$m$ in $\Omega$ leads to  a total of 
\begin{equation}\label{equ:Tslow}
T_\text{samp} = m(p+1)N
\end{equation}
evaluations of the stochastic response function $f$. Such a cost is prohibitive
in many applications where $N$ might be  of the order of tens of thousands.

%As illustrated by Example~\ref{ex:synthetic}, the distribution of the indices
%$\{S_k(\omega),T_k(\omega)\}_{k=1}^p$ may contain significant information needed for sensitivity
%analysis (the distribution of a general $S_u(\omega)$ may be valuable but we focus on first order and total indices for simplicity). Computing these sensitivity indices is, in general,  costly. For
%instance, the computation of the $S_k(\omega)$'s and $T_k(\omega)$'s through the sampling based
%method from \cite{saltelli} with $N$ Monte Carlo samples  requires $(p+2)N$
%evaluations of $f(\vec{X}(\cdot),\omega)$ for each fixed $\omega \in \Omega$.
%To characterize the statistical properties of the  indices, an additional Monte
%Carlo sampling over $\Omega$ has to be performed.  Assuming a sample size of
%$m$ in $\Omega$ leads to  a total of 
%\begin{equation}\label{equ:Tslow}
%T_\text{samp} = m(p+2)N
%\end{equation}
%evaluations of the stochastic response function $f$. Such a cost is prohibitive
%in many applications where $N$ might be  of the order of tens of thousands.

We use surrogate models to reduce the cost. Namely, for each fixed $\omega$, we
construct  $\hat{f}(\vec{X}, \omega) \approx f(\vec{X}, \omega)$.  The
construction of the surrogate $\hat{f}$ requires an ensemble of function
evaluations, $\{ f(\vec{X}_j, \omega) \}_{j=1}^{n}$, where
$\{\vec{X}_j\}_{j=1}^n$ are realizations of the uncertain parameters, drawn
from the distribution law $F_\vec{X}$. 

The construction of an efficient surrogate only requires $n$ function evaluations where $n$ is much smaller than the number of 
Monte Carlo samples, i.e., $n \ll N$. In addition, and as noted
earlier, most surrogates allow inexpensive or even analytic calculation of the
Sobol' indices. This reduces the total cost from~\eqref{equ:Tslow} to 
\begin{equation}\label{equ:Tfast}
T_\text{surrogate} =  m n
\end{equation}
evaluations of $f$, where $n$ is significantly smaller than $(p+1)N$.  
While the computational cost~\eqref{equ:Tfast} appears to be independent of the
uncertain parameter dimension $p$,  it should be noted that the choice of $n$ depends on $p$. This dependence 
is linked to the surrogate model itself.  An
adaptive surrogate model such as MARS can exploit the problem structure and
thus tempers this dependence.

Higher order indices may also be computed and similar cost analysis may be done. We summarize the main steps of our method for computing general sensitivity indices for
stochastic models in Algorithm~\ref{alg:method}.

\begin{algorithm}
\caption{Efficient approximation of $\{S_u\}_{u \subset \{1,2,\dots,p\}}$}
\begin{algorithmic}
\For{$i$ from $1$ to $m$}
\State Randomly generate $\omega_i \in \Omega$
\State Generate realizations $\{\vec{X}_j\}_{j=1}^n$ of the uncertain parameter vector 
\State Evaluate $f(\vec{X}_j, \omega_i)$, for each $j = 1, \ldots, n$
\State Construct surrogate $\hat{f}(\vec{X},\omega_i) \approx f(\vec{X},\omega_i)$ using 
data set $(\vec{X}_j, f(\vec{X}_j, \omega_i))_{j=1}^n$
\State Compute $\hat{S}_u(\omega_i)$ using the surrogate $\hat{f}(\vec{X},\omega_i)$ for $u \subset \{1,2,\dots,p\}$
\EndFor
\State Approximate statistical properties of $S_u$ using $\{\hat{S}_u(\omega_i)\}_{i=1}^m$ for  $u \subset \{1,2,\dots,p\}$
\end{algorithmic}
\label{alg:method}
\end{algorithm}

The algorithm returns $m$ realizations of Sobol' indices of $\hat{f}(\vec{X},\omega)$, i.e., $\hat S_u(\omega)$,  $u \subset \{1,2,\dots,p\}$.  
Let us denote these realizations $\hat{S}_u^i \stackrel{iid}{\sim} \hat{S}_u$, 
$i = 1, \ldots, m$,  $u \subset \{1,2,\dots,p\}$ and consider the sample $r$-th moment
\begin{equation}
\hat{\mu}_u^{[r]}(\omega) = \frac{1}{m} \sum_{i=1}^m (\hat{S}_u^i(\omega))^r. \label{mudef}
\end{equation}
Clearly, we have
\[
\mathbb E_{\omega}\{\hat{\mu}_u^{[r]}\} =\mathbb E_{\omega}\{(\hat{S}_u)^r\}
\quad
\text{and}
\quad
\operatorname{Var}_\omega\{\hat{\mu}_u^{[r]}\}=\frac{\operatorname{Var_\omega}\{(\hat{S}_u)^{r}\}}{m},
\quad u \subset \{1,2,\dots,p\}.
\]
The error in approximating $\mathbb E_{\omega}\{(S_u)^r\}$
can be decomposed into Monte Carlo error using
$m$ samples from $\Omega$ and surrogate approximation error using $n$ samples from $\Theta$.

\begin{proposition}\label{prp:estimates}
Let $\hat{\mu}_u^{[r]}$, $S_u$, and $\hat{S}_u$ be as defined above. Then, 
\begin{enumerate}
\item $\displaystyle \mathbb E_{\omega}\{\hat{\mu}_u^{[r]}-\mathbb E_{\omega}\{(S_u)^r\} \}=\mathbb E_{\omega}\{(\hat{S}_u)^r-(S_u)^r\}$,
\item $\displaystyle 
\operatorname{Var}_\omega\{\hat{\mu}_u^{[r]}-\mathbb E_{\omega}\{(S_u)^r\}\} \le 
   {\mathbb E_{\omega}\{(\hat{S}_u)^r\}(1-\mathbb E_{\omega}\{(\hat{S}_u)^r\})}/{m} \le 
   \frac 1{4m}$.
\end{enumerate}
\end{proposition}

\begin{proof}
The first statement follows from
\begin{multline*}
\mathbb E_{\omega}\{\hat{\mu}_u^{[r]}-\mathbb E_{\omega}\{(S_u)^r\}\} = 
\mathbb E_{\omega}\{\hat{\mu}_u^{[r]}-\mathbb E_{\omega}\{(\hat{S}_u)^r\}+\mathbb E_{\omega}\{(\hat{S}_u)^r\}-\mathbb E_{\omega}\{(S_u)^r\}\}\\
=\left(\mathbb  E_{\omega}\{\hat{\mu}_u^{[r]}\}-\mathbb E_{\omega}\{(\hat{S}_u)^r\} \right)+\mathbb E_{\omega}\{(\hat{S}_u)^r-(S_u)^r\}
=\mathbb E_{\omega}\{(\hat{S}_u)^r-(S_u)^r\}.
\end{multline*}

For the second statement, we note
\[
\begin{aligned}
\operatorname{Var}_\omega\{\hat{\mu}_u^{[r]}-\mathbb E_{\omega}\{(S_u)^r\}\} &= \operatorname{Var}_\omega\{\hat{\mu}_u^{[r]}\}\\
&= \frac{\operatorname{Var}_\omega\{(\hat{S}_u)^r\}}{m}
  \le \frac{\mathbb E_{\omega}\{(\hat{S}_u)^r\}(1-\mathbb E_{\omega}\{(\hat{S}_u)^r\})}{m}
  \le \frac{1}{4m},
\end{aligned}
\]
where the inequalities follow from the Theorem 2 in
\cite{inequal} and the fact that $\hat{S}_u$ is
supported on $[0,1]$. 
\end{proof}

To understand the implication of the above result,  consider the
point estimator for the expected value of $S_u(\omega)$ given by the sample mean:
\begin{equation}\label{equ:sample_avg}
   \hat{\mu}_u^{[1]}(\omega) = \frac{1}{m} \sum_{i=1}^m \hat{S}_u^i(\omega).
\end{equation}
Proposition~\ref{prp:estimates} characterizes the bias of this estimator
as the approximation error due to the surrogate, i.e., $\mathbb E_{\omega}\{\hat{S}_u-S_u\}$.
Further, since,  $\mathbb E_{\omega}\{(\hat{S}_u)^r\} \in [0,1]$, for every $r \ge 1$,
 the second statement of Proposition~\ref{prp:estimates} indicates that even a modest value of $m$, 
say in the order of a few hundreds, can be very effective in obtaining an estimator with 
small variance. Finally,   an estimate of the error can be obtained in the $L^2$ norm by using the elementary definition
of the variance and Proposition~\ref{prp:estimates}
\[
\mathbb E_{\omega}\{(\hat{\mu}_u^{[1]}-\mathbb E_{\omega}\{S_u\} )^2\} \le \mathbb E_{\omega}\{\hat{S}_u-S_u\}^2 + \frac 1{4m}.
\]

\section{Numerical results}
\label{sec:numerical}

The following two examples illustrate some of the points raised in the previous section:
(i) the convergence of the estimators  as a function of $n$
(number of samples to build  MARS) and $m$ (number of samples over $\Omega$) and 
(ii) the effect of the surrogate bias on the
statistical distribution of the Sobol' indices.

\subsection{The stochastic g-function} \label{sec:gfunc}
Let $(\Theta,\mathcal{E}, \lambda)$ and  $(\Omega, \mathcal{F}, \nu)$ be two probability spaces and let $\mathbf X:\Theta \rightarrow \mathbb R^{15}$ and $W:\Omega \rightarrow \mathbb R$ be two random variables such that 
\begin{eqnarray*}
&&\mathbf X = [X_1, \dots, X_p] \mbox{ with } X_i  \stackrel{iid}{\sim} \mathcal U(0,1), i = 1,\dots, p, \\
&&W \sim \operatorname{Beta}(5,3); 
\end{eqnarray*}
in other words, we have
\begin{eqnarray*}
&&\lambda\bigl(\mathbf X \in (c_1,d_1)\times \cdots \times (c_{15},d_{15})\bigr) = \prod\limits_{k=1}^{15} (d_k-c_k) \mbox{ for any } 0 \le c_k\le d_k\le 1, 
k=1,\dots, p,\\
&&\nu(W \in (a,b)) = K \int_a^b   t^{4}(1-t)^{2} \, dt \mbox{ for any $a$, $b$},  0 \le a \le b \le 1,
\end{eqnarray*}
with a normalization factor  
$K = \Gamma(5+3)/\big(\Gamma(5)\Gamma(3)\big) = 105$.

We now define a stochastic version of the g-function 
\begin{eqnarray}
f(\mathbf X(\vartheta),\omega)=\prod\limits_{k=1}^{15} \frac{|4X_k(\vartheta)-2|+a_k(W(\omega))}{1+a_k(W(\omega))}, \label{sgfunction}
\end{eqnarray}
where $X_k(\vartheta)$ is the $k^{th}$ component of $\mathbf X(\vartheta)$ and the parameters $a_k:
[0,1] \to \mathbb R$, $k = 1,\dots, 15$, are chosen to create a
variety of means and variances for the Sobol' indices. Analytic expressions for
the $a_k$'s are given in Table~\ref{tab:aj}.

\begin{table}[ht]
\centering
\ra{1.3}
\begin{tabular}{lll}
\toprule
$a_1(t)=(1-t)^5$ & $a_2(t)=t^5$ & $a_3(t)=\sin^2(8t)$\\
$a_4(t)=\sin^2(10(1-t))$ & $a_5(t)=\cos^2(10(1-t))$ & $a_6(t)=\cos^2(8t)$\\
$a_7(t)=(1.5-t)^2$ & $a_8(t)=(.5+t)^2$ & $a_9(t)=(3-t)^2$\\
$a_{10}(t)=(2+t)^2$ & $a_{11}(t)=(3.5-t)^2$ & $a_{12}(t)=(2.5+t)^2$\\
$a_{13}(t)=(4-t)^2$ & $a_{14}(t)=(3+t)^2$ &  $a_{15}(t)=(4+t)^2$\\
\bottomrule
\end{tabular}
\caption{Expressions of the parameters $a_k$, $k=1, \ldots, 15$, for the stochastic g-function example (\ref{sgfunction}).}
\label{tab:aj}
\end{table}

We compute the $S_k$'s analytically and subsequently evaluate $\mathbb E_{\omega}\{S_k\}$, $k = 1, \dots, 15$, 
using numerical quadratures. The trapezoidal rule with $10^6$
quadrature nodes is used to ensure accurate computation of the expectations. 

\begin{figure}[ht]
\centering
\includegraphics[width=.49\textwidth]{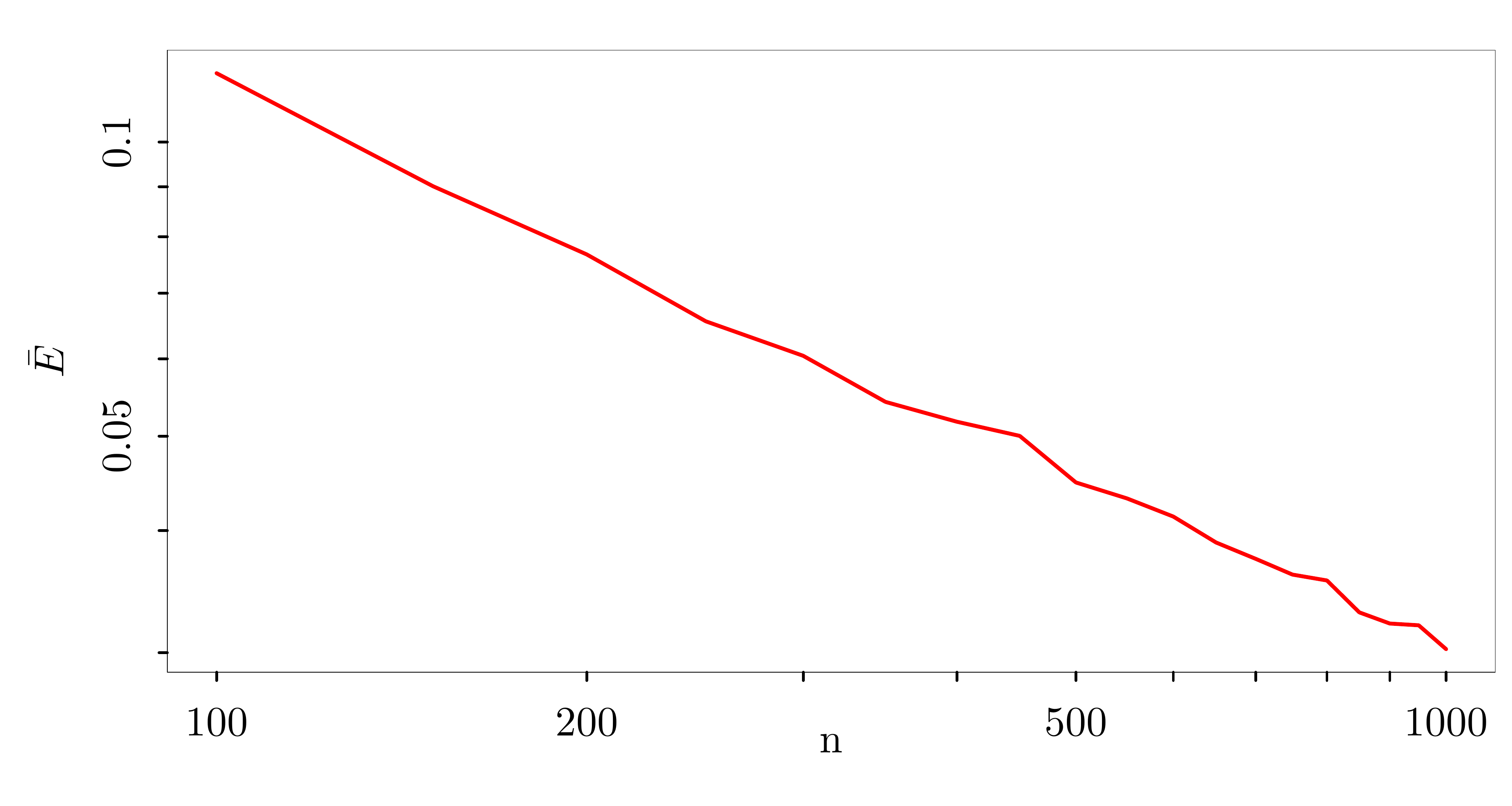}
\includegraphics[width=.49\textwidth]{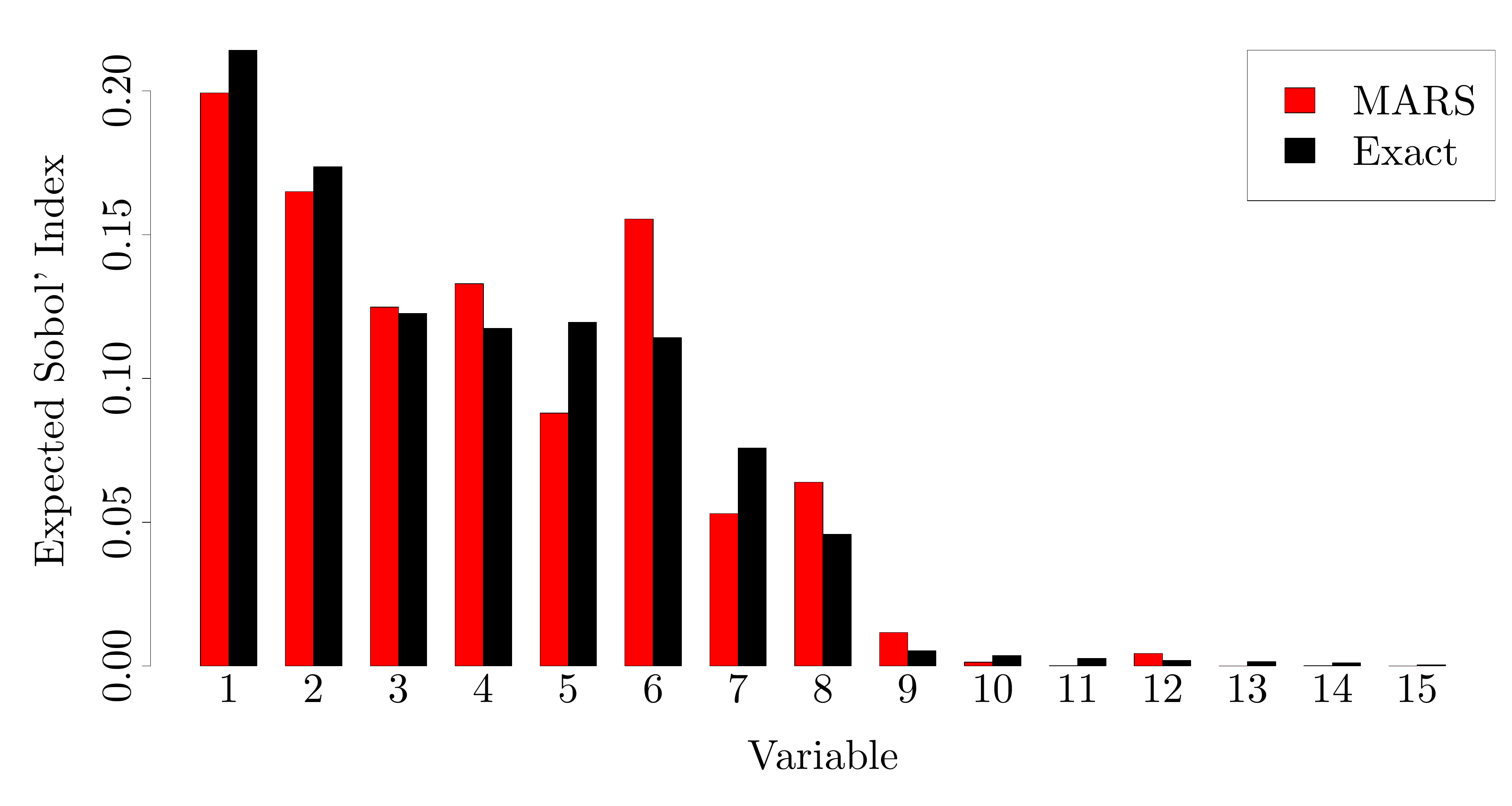}
\caption{Convergence of the expected Sobol' indices for the stochastic g-function (\ref{sgfunction}). Left: average error $\bar E$ (\ref{errordefstoch}) in the normalized expectation  of the indices as the surrogate sampling size $n$ varies; the  empirical convergence rate is 0.59. Right:
comparison of normalized expectation of the  indices for a representative sample size of   $n=600$.
}
\label{fig:varydata_barplot}
\end{figure}

In our first test, we study the error approximating $\mathbb E_{\omega}\{S_k\}$, $k =1,\dots,15$,  
as a function of the number of surrogate samples
$n$. For $n=100,150,\dots,950,1000$, the $\hat{\mu}_k^{[1]}$'s from (\ref{equ:sample_avg}) are obtained from Algorithm~\ref{alg:method}. Samples from $\Theta$ are taken using a Latin Hypercube design. To remove dependence upon sampling, we generate $500$ different datasets for each fixed $n$ and define the error as the average errors over these datasets
\begin{eqnarray}
\bar{E}=\frac{1}{500} \sum_{\ell=1}^{500} \left\| \frac{\mathbb E_{\omega}\{\mathbf S\}}{\sum\limits_{k=1}^p \mathbb E_{\omega}\{S_k\}}-
\frac{\hat{\bm \mu}(\ell,n)}{\sum\limits_{k=1}^p \hat{\mu}_k(\ell,n)}\right\|_{\infty},  \label{errordefstoch}
\end{eqnarray}
where
$\hat{\bm\mu}(\ell,n)=[\hat{\mu}_{1}(\ell,n),\hat{\mu}_{2}(\ell,n),\dots,\hat{\mu}_{p}(\ell,n)]$ and each $\hat{\mu}_{k}(\ell,n)$
is a realization of the random variable $\hat{\mu}_k^{[1]}$ using the $\ell^{th}$ dataset of size $n$. 
Figure~\ref{fig:varydata_barplot} (left) shows the error $\bar E$ as a function
of $n$ while Figure~\ref{fig:varydata_barplot} (right) compares the normalized
expected Sobol' indices of MARS with the normalized exact indices for a representative sample of size $n=600$.
We study the effect of $m$ in~\eqref{equ:sample_avg} in Figure~\ref{fig:conv_in_m}, which shows the
convergence of $\hat{\mu}_1^{[1]}$ and $\hat{\mu}_3^{[1]}$ as the number of
samples $m$ increases. The results in Figure~\ref{fig:conv_in_m} are computed using the sample
of size $n=600$ from Figure~\ref{fig:varydata_barplot} (right). The first and
third variables are chosen because they have the largest expectation and
variance, respectively. 
These results confirm both the efficiency of MARS as a
surrogate and the fast convergence of the expectation of the indices with only
$200$ samples in $\Omega$.

\begin{figure}[ht]
\centering
\includegraphics[width=.65 \textwidth]{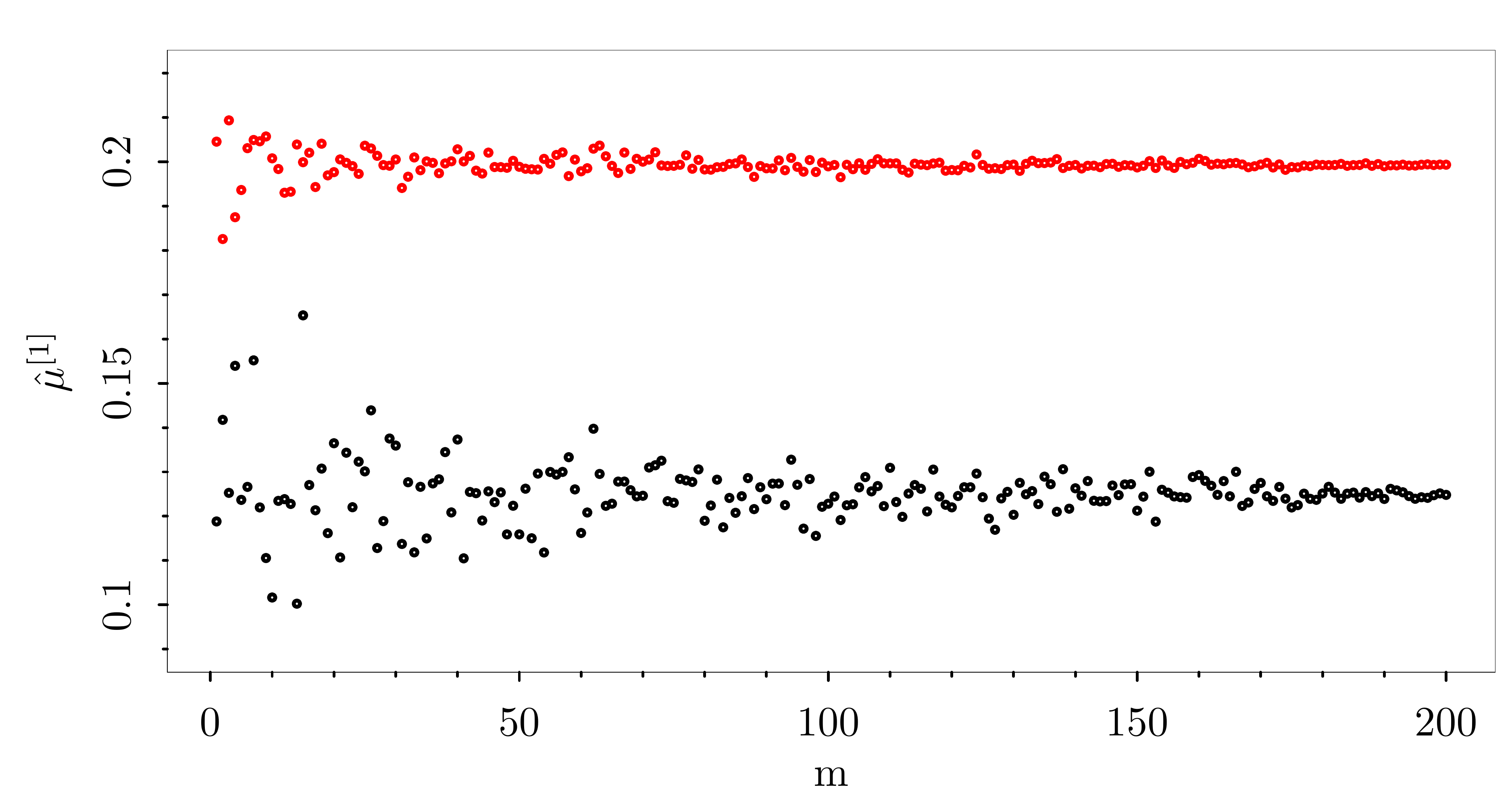}
\caption{Dependency  of $\hat{\mu}_1^{[1]}$ (red) and $\hat{\mu}_3^{[1]}$ (black) on the number of samples $m$ in $\Omega$. 
} 
\label{fig:conv_in_m} 
\end{figure}

Accurate approximations of the distributions of the Sobol' indices can be obtained from sampling their analytical expressions; as above, we take $10^6$ samples from $\Omega$. We use these highly accurate approximations to
 assess convergence in distribution of the Sobol's indices computed through our proposed method.
 \begin{figure}[ht]
\centering
\includegraphics[width=.49\textwidth]{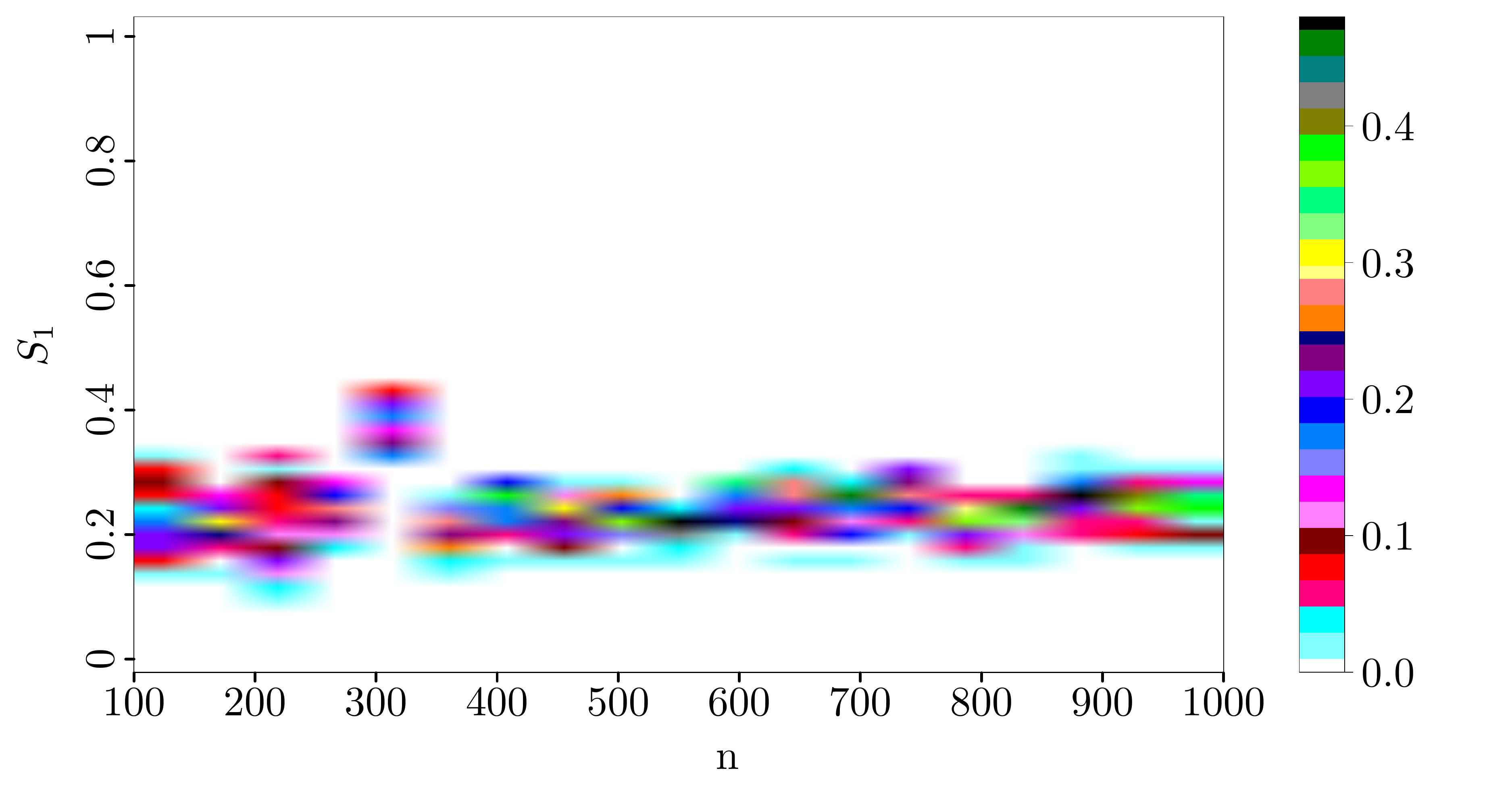}
\includegraphics[width=.49\textwidth]{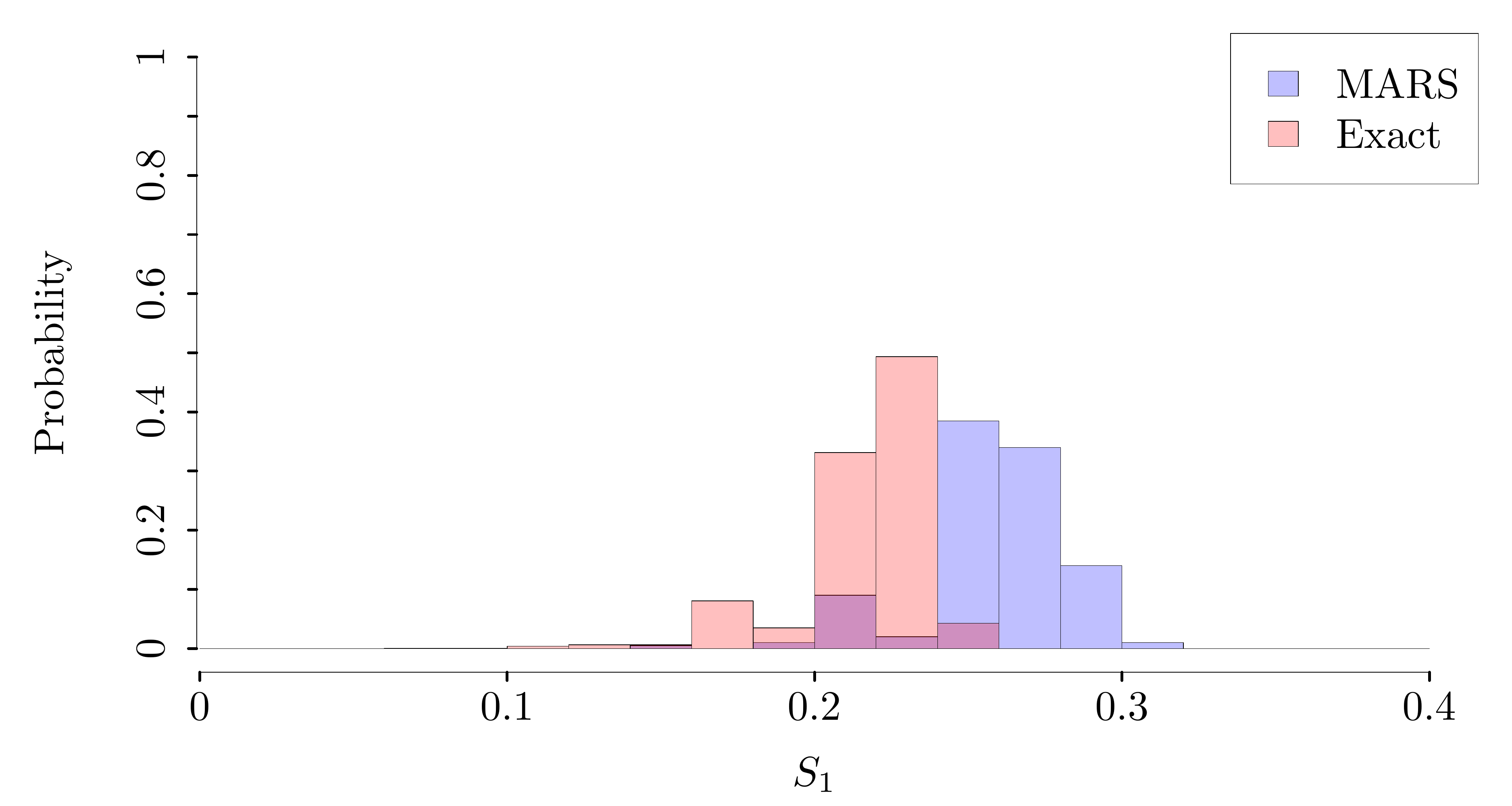} \\
\includegraphics[width=0.49\textwidth]{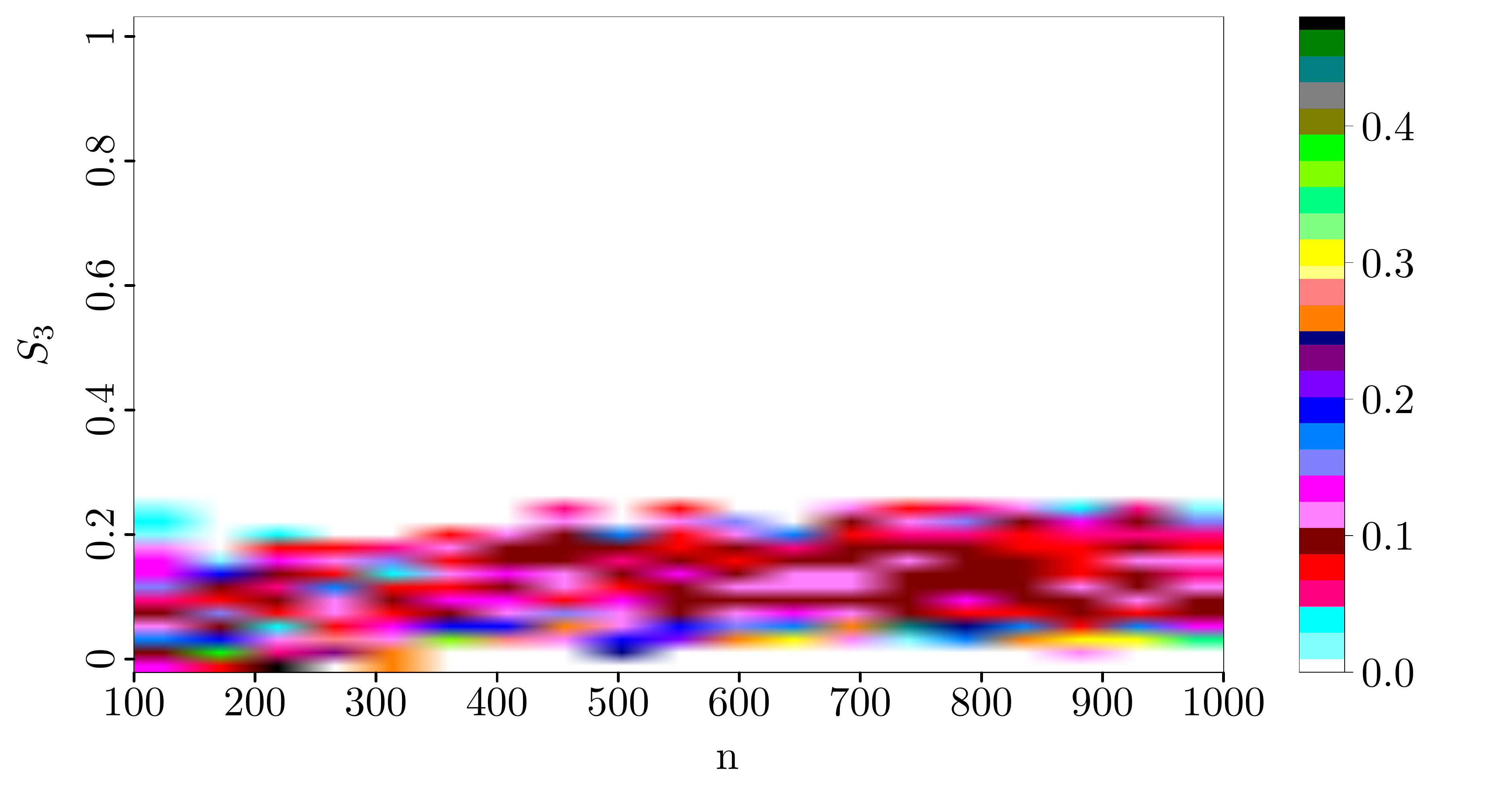}
\includegraphics[width=0.49\textwidth]{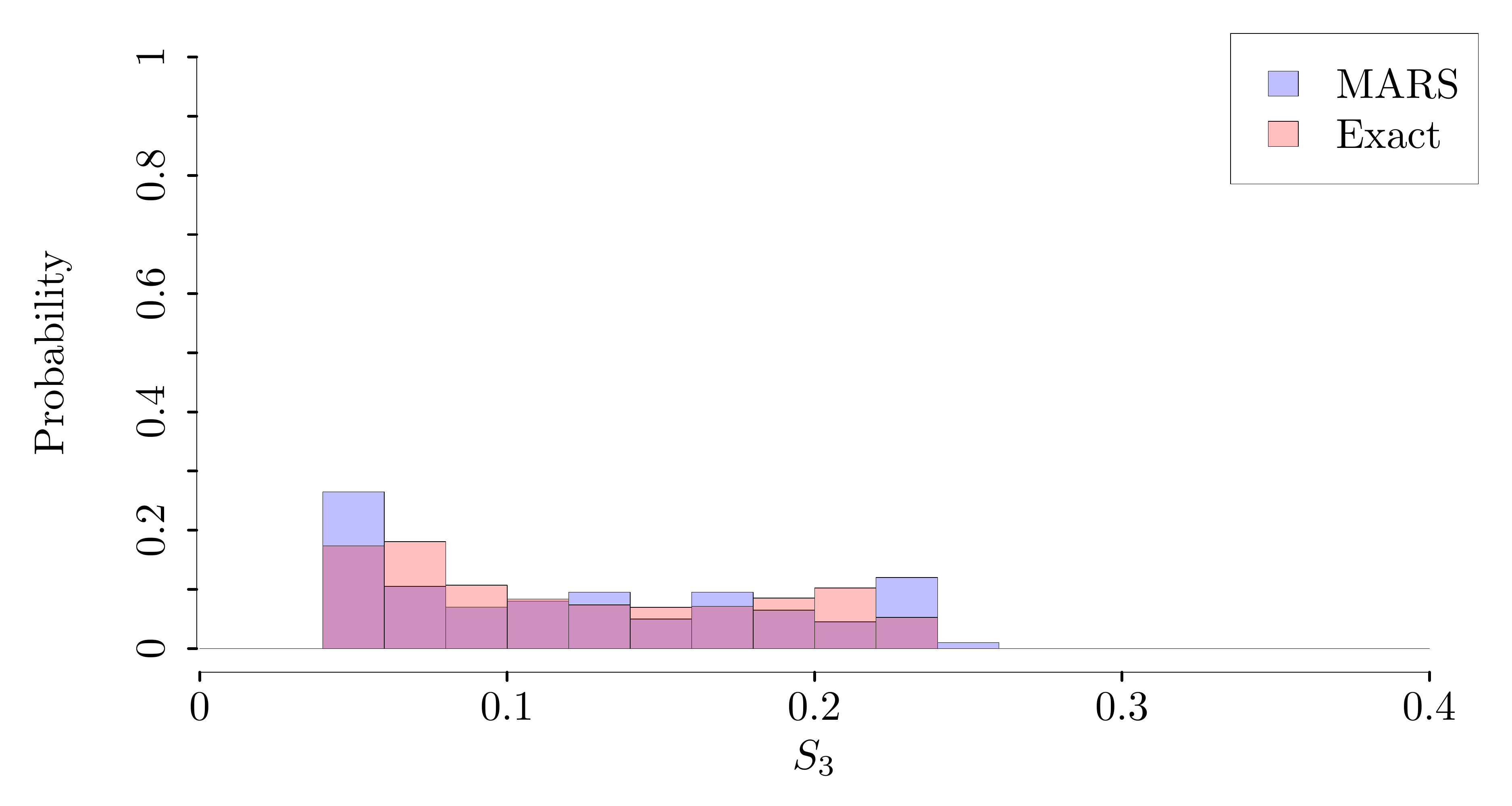}
\caption{Convergence in distribution of the Sobol' indices for the g-function (\ref{sgfunction}). Top row; $S_1$, index with the largest expectation; bottom row: $S_3$, index with the largest variance. Left: heat map of the histograms as the surrogate sampling size $n$ varies. Each vertical slice is a histogram for a fixed $n$; right: comparison of the ``exact" (see text) and approximation distributions using $n=1000$.}
\label{fig:distconv1fig}
\end{figure}

For each $\omega_i$, $i=1,\dots, 200$, one thousand points are sampled from the
uncertain parameter space; these are subsampled for $n=100,150,\dots,950,1000$
and the resulting histograms are evaluated.
Figure~\ref{fig:distconv1fig} illustrates convergence in distribution of both
$S_1$, the index with largest expectation, and $S_3$, the index with largest
variance.
\begin{figure}[h]
\centering
\includegraphics[width=.5 \textwidth]{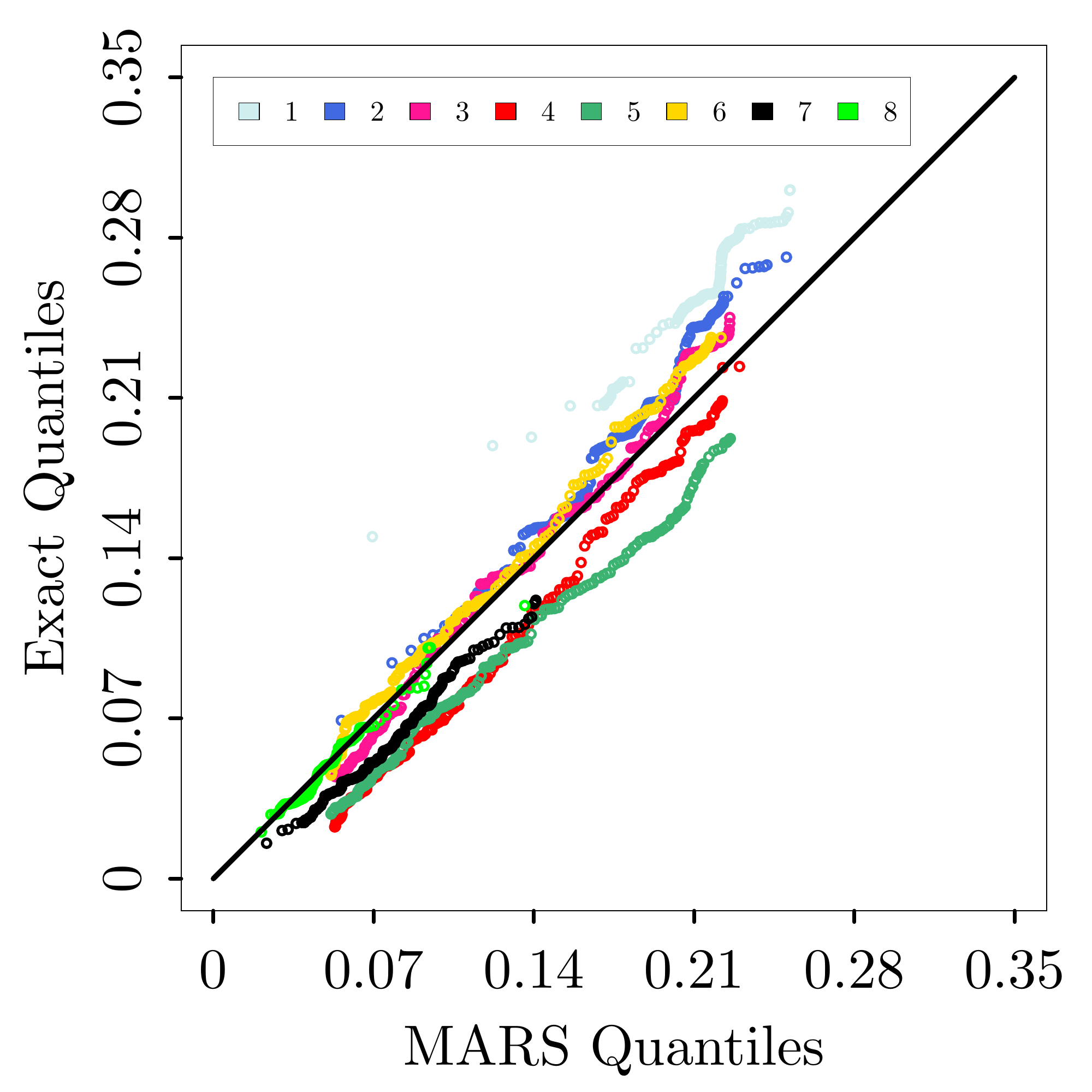}
\caption{QQ plot of the Sobol' indices of the eight most important variables. Lying above or below the line indicates being biased high or low respectively.}
\label{fig:qq}
\end{figure}
For $S_1$, the histograms appear to converge, see
Figure~\ref{fig:distconv1fig}, top left; however,
Figure~\ref{fig:distconv1fig}, top right, shows that they converge to a
distribution that is biased high. This results from the built-in adaptivity of
MARS which causes an inherent bias toward the more important variables. As
mentioned in Section 1, this is a useful feature for the purpose of dimension
reduction.  The bottom row of Figure~\ref{fig:distconv1fig} illustrates the
unbiased convergence in distribution of $S_3$. Figure~\ref{fig:qq} further
illustrates the bias from MARS with a QQ plot of the eight most important
variables. The ``exact" distribution was generated using $10^6$ samples from the analytic expressions of the Sobol' indices. Lying along the line indicates being unbiased; lying below or above
the line indicates being biased low or high respectively. This plot
demonstrates the general trend that the most important variables are biased high
while the less important variables are biased low; $S_6$ and $S_8$ fail to follow this trend.

\subsection{Genetic oscillator}\label{sec:genosci}

We apply the proposed  stochastic sensitivity analysis method to the study of a circadian oscillator mechanism
from biochemistry. 
The problem  is detailed in~\cite{VilarEtAl02} and is commonly referred to as the genetic
oscillator. It  corresponds to a biochemical reaction network consisting of nine species
and sixteen reactions. 
The nine system species are described in Table~\ref{tab:species}.

\begin{table}[ht]
\centering
\ra{1.3}
\begin{tabular}{ll}
\toprule
$\snDA$, $\snDAp$ & activator genes \\
$\snDR$,  $\snDRp$ & repressor genes \\
$\snA$, $\snR$ & activator and repressor proteins \\
$\snMA$, $\snMR$ & mRNA of $\snA$ and $\snR$ \\
$\snC$ &  complex species \\
\bottomrule
\end{tabular}
\caption{Nine species of the genetic oscillator problem from \cite{VilarEtAl02}.}
\label{tab:species}
\end{table}

Denoting the number of molecules of
each of the species with the corresponding symbol, the state vector of the
system is given by 
\[
[\snDA, \snDAp, \snA, \snDRp, \snDR, \snR, \snMA, \snMR,
\snC] \in \R^9.  
\]
The initial state is taken as 
\[
[1, 0, 0, 0, 1, 0, 0, 0, 0],
\]
 that is, initially, $\snDA = \snDR = 1$, and all the other variables are set to zero.

The reactions and reaction rates are listed in
Table~\ref{tbl:reactions} which also includes the nominal reaction rates from~\cite{VilarEtAl02}.  
\begin{table}\centering
\ra{1.3}
\begin{tabular}{clll}
\toprule
reaction \# & reaction && rate (nominal value) \\
\midrule
1 & $\snDA + \snA $&$\to \snDAp$         & $\gamma_A$        (1.0) \\
2 & $\snDAp      $&$\to \snDA + \snA$    & $\theta_A$        (50.0) \\
3 & $\snDR + \snA $&$\to \snDRp$         & $\gamma_R$        (1.0) \\
4 & $\snDRp      $&$\to \snDR + \snA$    & $\theta_R$        (100.0) \\
5 & $\snA + \snR  $&$\to \snC$           & $\gamma_C$        (2.0)\\
6 & $\snDA       $&$\to \snDA + \snMA$   & $\alpha_A$        (50.0) \\
7 & $\snDR       $&$\to \snDR + \snMR$   & $\alpha_R$        (0.01) \\
8 & $\snDAp      $&$\to \snDAp + \snMA$  & $\alpha_A^\prime$ (500.0) \\
9 & $\snDRp      $&$\to \snDRp + \snMR$  & $\alpha_R^\prime$ (50.0) \\
10 & $\snMA      $&$\to \snMA + \snA$    & $\beta_A$         (50.0) \\
11 & $\snMR      $&$\to \snMR + \snR$    & $\beta_R$         (5.0) \\
12 & $\snMA      $&$\to \emptyset$     & $\delta_{MA}$       (10.0)\\
13 & $\snA       $&$\to \emptyset$     & $\delta_A$          (1.0)\\
14 & $\snMR      $&$\to \emptyset$     & $\delta_{MR}$       (0.5) \\
15 & $\snR       $&$\to \emptyset$     & $\delta_R$          (0.2) \\
16 & $\snC       $&$\to \snR$           & $\delta_A$         (same as react. 13)  \\
\bottomrule
\end{tabular}
\caption{Reactions and reaction rates for the genetic oscillator system~\cite{VilarEtAl02}.} 
\label{tbl:reactions}
\end{table}
We consider parametric uncertainties in the reaction rate constants; that is, 
the uncertain parameter vector for the system is given by 
\[\vec{X} = [\gamma_A,\theta_A,\gamma_R,\theta_R,\gamma_C,
\alpha_A,\alpha_R,\alpha_A^\prime,\alpha_R^\prime,\beta_A,\beta_R,\delta_{MA},\delta_A,\delta_{MR},
\delta_R] \in  \R^{15}.
\]
We assume that the coordinates of $\vec{X}$, i.e., the reaction rates, are iid uniform
random variables centered at their respective nominal values given in
Table~\ref{tbl:reactions}, and with a $10\%$ perturbation around the mean.  

With fixed reaction rates the time evolution of the state vector is stochastic.
The sequence of reactions is random with probabilities parameterized by the
reaction rates and state vector; see e.g.,~\cite{SamadKhammashPetzold05}. We compute realizations of the genetic
oscillator through Gillespie's stochastic simulation algorithm
(SSA)~\cite{ssa1,ssa2,SamadKhammashPetzold05}. To fix a realization of the inherent stochasticity and
sample the reaction rates we generate and save a sequence of random numbers to
input to SSA for each sample of the reaction rates. This corresponds to
evaluating $f(\vec{X}_j, \omega_i)$, for each $j = 1, \ldots, n$, in
Algorithm~\ref{alg:method}; $\omega_i$ corresponds to the fixed sequence of
random numbers.

Our goal is to evaluate  the sensitivity of the number of $\snC$ molecules to the uncertain
reaction rates.  
Letting the probability spaces $(\Theta, \mathcal{E}, \nu)$ and $(\Omega, \mathcal{F}, \mu)$ carry the intrinsic and the parametric
randomness of the system, respectively, we  
note that 
$\snC: \Theta \times \Omega
\times [0, T_\text{final}] \to \R$ is a stochastic process.
 Figure~\ref{fig:Cplot} illustrates the
dynamics of $\snC = \snC(\vartheta,\omega, t)$ by displaying
 four typical realizations of the stochastic process when the uncertain
parameters are fixed at their nominal values.  The differing periods and small
oscillations are a result of the inherent stochasticity of the system.
\begin{figure}[h]
\centering
\includegraphics[width=.75\textwidth]{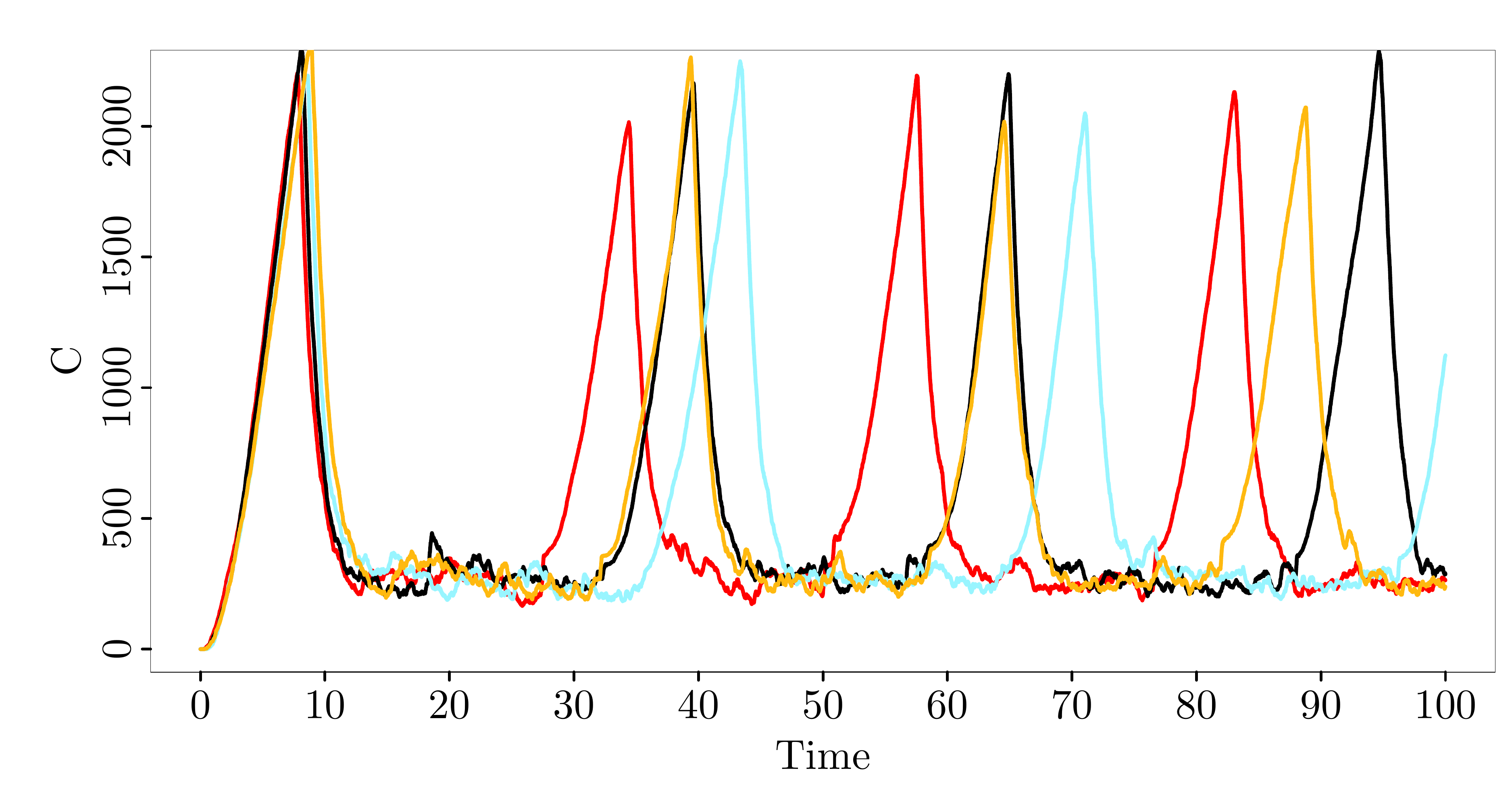}
\caption{Genetic oscillator: four realizations of the evolution of the  complex $C$.}
\label{fig:Cplot}
\end{figure}

We use Algorithm~\ref{alg:method}, with $n=400$, $m=200$, and 
MARS as the surrogate model. By  subsampling $\Theta$ from our
existing data, we assess convergence in $n$ and determine that
$n=400$ is an adequate sample size for this application. Moreover, 
thanks to Proposition~\ref{prp:estimates}, a relatively small $m$ is
sufficient to ensure a small variance of the estimators of the Sobol' indices.

Figure~\ref{fig:genoscillsensitivity}, left, shows the time
evolution of the expectation of the Sobol' indices; for each index, the expectation 
 becomes periodic after an initial transient.  Moreover, we observe that the reaction rates
$\beta_R$ and $\alpha_R^\prime$ have the most notable contribution to the model
variance during the transient regime.  After the transient regime, the
degradation rates for the proteins $\snA$ and $\snR$, i.e.,  $\delta_A$ and
$\delta_R$,  are the most important factors.  Figure~\ref{fig:genoscillsensitivity}, top right,  displays the time
evolution of the expectation of the Sobol' indices for these two 
reaction rates. We note in particular that $\delta_A$
and $\delta_R$ periodically swap  role of the most dominant contributor to
variance of $\snC$ after the intial transient regime.  
The time-dependent behavior of the statistical distribution of sensitivity
index for $\delta_R$ is illustrated in Figure~\ref{fig:genoscillsensitivity}, bottom right.
\begin{figure}
\begin{minipage}[t]{.5\linewidth}\vspace{0pt}
\includegraphics[width=\linewidth]{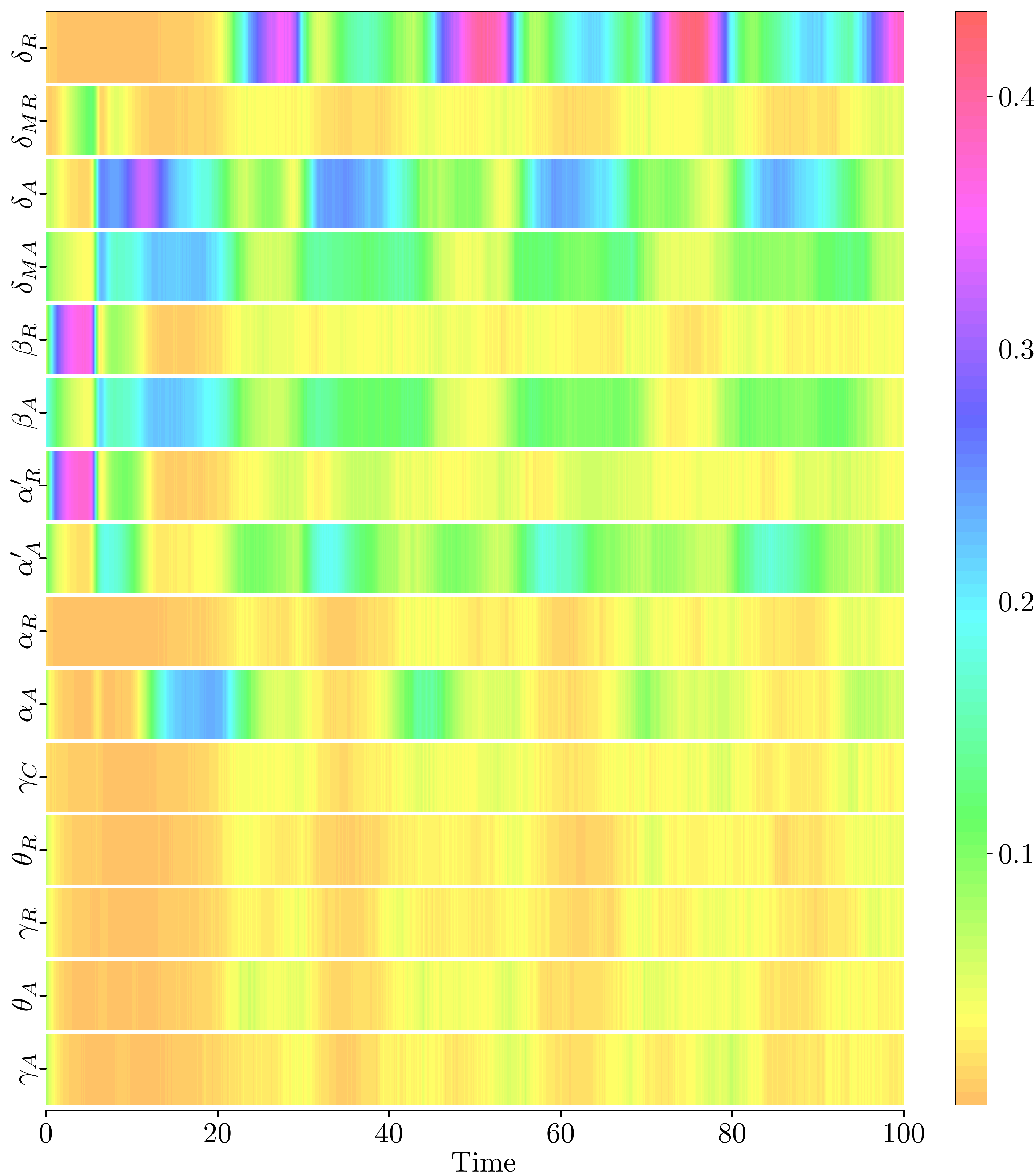}
\end{minipage}
\begin{minipage}[t]{.5\linewidth}\vspace{0pt}\raggedright
\vspace{0.2cm}
\begin{tabular}{l}
\hspace{0.04cm}
\includegraphics[width=.855\textwidth,height=.5\textwidth]{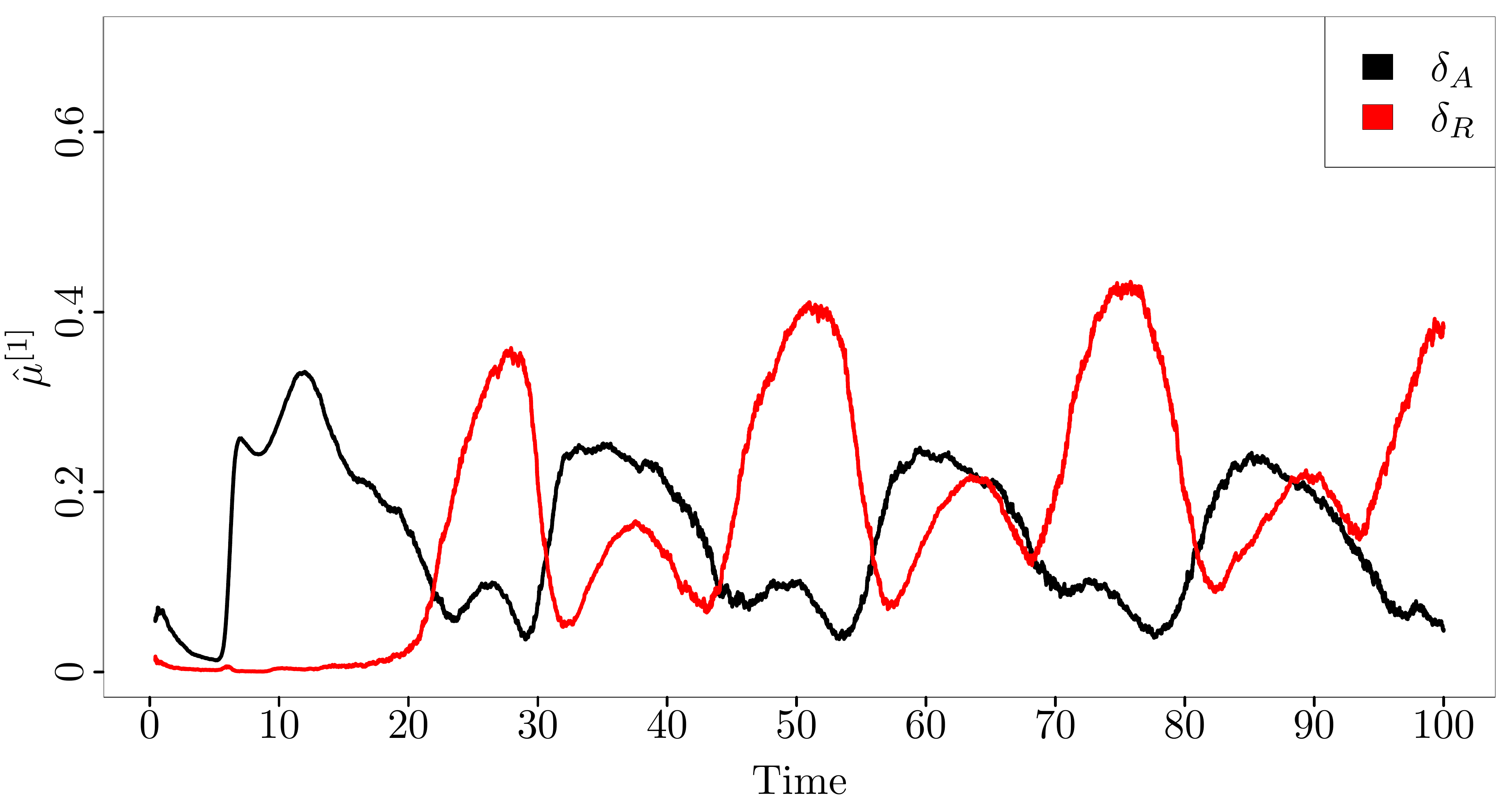}\\
\includegraphics[width=1\textwidth]{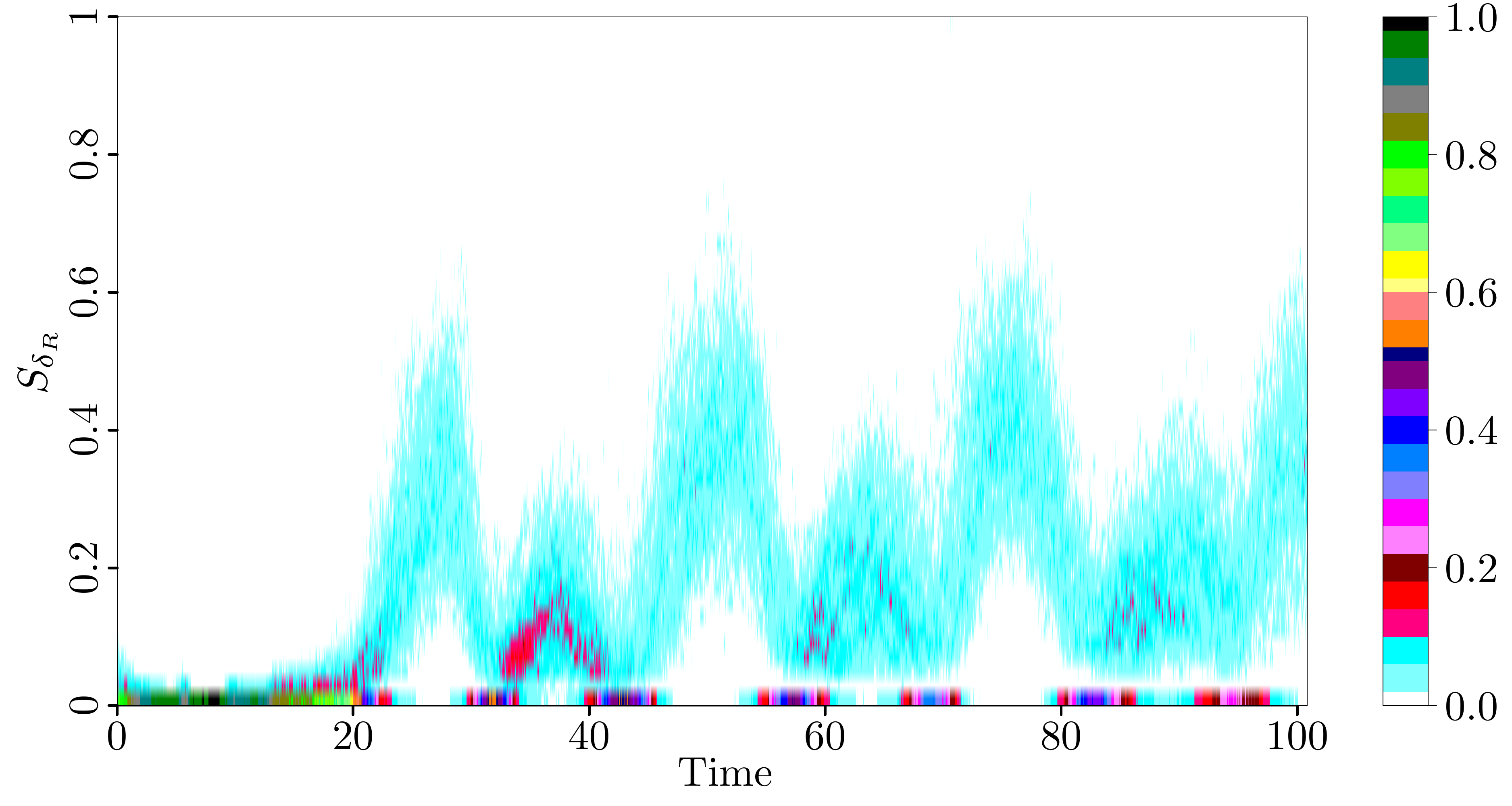}
\end{tabular}
\end{minipage}
\caption{ Evolution of the Sobol' indices for the genetic oscillator. Left: expectation of the
Sobol' indices. Each row corresponds to a specific reaction rate. Top right: time
evolution of the expectation of the Sobol' indices for the two most important
reaction rates.  Bottom right:  time evolution of the histogram of the Sobol'
index of $\delta_R$. Each vertical slice is a histogram for the Sobol' index of
$\delta_R$ at a given time. 
}
\label{fig:genoscillsensitivity}
\end{figure}

The computational cost of the above analysis is $n \times m = 400 \times 200 =
80,000$ SSA simulations. While this is still a significant number of function
evaluations, it should be contrasted with the complexity of traditional sampling
based method for computing the Sobol' indices, given in~\eqref{equ:Tslow},
which, for such a stochastic model, would be orders of magnitudes larger.

\section{Summary and future work}
\label{sec:discussion}

We have proposed and investigated a strategy for the global sensitivity
analysis of stochastic models.  Within this framework,  a thorough analysis  of
variable importance is obtained by computing  the statistical properties  of
the Sobol' indices.  The proposed approach requires sampling in the product of
probability spaces carrying the model stochasticity and parametric uncertainty.
The number of samples in the uncertain parameter space is driven by the choice
of surrogate model construction, and can be controlled via the use of adaptive
surrogates such as MARS.  We provide theoretical and numerical evidence that
the moments of the indices can be evaluated with only a modest number of
samples in the stochastic space. 

As mentioned in the introduction, one may also consider performing sensitivity
analysis on $\mathbb E_\omega \{f(\vec{X},\omega) \}$ directly. However, as
illustrated in Example~\ref{ex:synthetic}, this approach may result in a
significant loss of information.  The
variance of the Monte Carlo estimator for $E_\omega \{f(\vec{X},\omega) \}$ is
$\operatorname{Var}_\omega \{f(\vec{X},\omega)\}/m$, which in general is not
known a priori. In contrast, the variance of the Monte Carlo estimator
for $\mathbb E_\omega \{ S_u(\omega) \}$ in our proposed framework is
bounded above by $1/(4m)$ independently of $f$.  This gives our proposed method
theoretical and computational advantages.

%Moreover, from computational point of view, this approach is not necessarily 
%cheaper than our proposed framework. 
%Thus our proposed framework guarantees that small
%values of $m$ may be taken; the magnitude of $m$ needed to compute $\mathbb
%E_\omega \{f(\vec{X},\omega) \}$ is not clear in general. 

Our numerical results focus on computing first order Sobol' indices and MARS
was shown to be an efficient surrogate for this task. Higher order indices may
be computed in our proposed framework as well, provided a sufficiently accurate surrogate
model is available.

%The efficiency of MARS for
%computing higher order indices is not well understood.
 
In our future work, we aim to address the following:

\begin{itemize}
\item Thorough analysis of the role played by the surrogates both in terms of acting 
as possible screening mechanisms (as is the case for MARS) and regarding index approximation errors.
\item Explore the use of other surrogates to compute higher order indices in our framework.
\item Convergence analysis of the distributions of the Sobol' indices and study of what can be inferred 
from them in light of surrogate induced biases. 
\end{itemize}

\bibliographystyle{siam}
\bibliography{stochsobol}

\end{document}

%% file: def.tex
\newcommand{\ra}[1]{\renewcommand{\arraystretch}{#1}}

\renewcommand{\vec}[1]{{\mathchoice
                     {\mbox{\boldmath$\displaystyle{#1}$}}
                     {\mbox{\boldmath$\textstyle{#1}$}}
                     {\mbox{\boldmath$\scriptstyle{#1}$}}
                     {\mbox{\boldmath$\scriptscriptstyle{#1}$}}}}

\newcommand{\R}{\mathbb{R}}

\newcommand{\snDA} {\mathrm{D_A}}
\newcommand{\snDR} {\mathrm{D_R}}
\newcommand{\snDAp}{\mathrm{D^\prime_A}}
\newcommand{\snDRp}{\mathrm{D^\prime_R}}
\newcommand{\snMA} {\mathrm{M_A}}
\newcommand{\snA}  {\mathrm{A}}
\newcommand{\snMR} {\mathrm{M_R}}
\newcommand{\snR}  {\mathrm{R}}
\newcommand{\snC}  {\mathrm{C}}

\newtheorem{example}{Example}[section]

\newcommand{\defeq}{\vcentcolon=}